%% file: rev2.tex
\documentclass[10pt,reqno]{amsart}

\input{header.tex}
\title[Two Proposals for Robust PCA]{Two Proposals for Robust PCA\\ using Semidefinite Programming}
\date{\today}
\author{Michael McCoy}
\author{Joel A. Tropp} 
\thanks{This work has been supported in part by ONR awards
  N00014-08-1-0883 and N00014-11-1-0025, AFOSR award FA9550-09-1-0643,
  and a Sloan Fellowship. This research was performed while the
  authors were in residence at IPAM.  The authors can be contacted via
  email at \texttt{\{mccoy,jtropp\}@acm.caltech.edu} or postal mail at
  Computing \& Mathematical Sciences, 1200 E. California Blvd., MC
  305-16, California Inst. Technology, Pasadena, CA 91125 }

\begin{document}

\begin{abstract}
  The performance of principal component analysis (PCA) suffers badly
  in the presence of outliers.  This paper proposes two novel
  approaches for robust PCA based on semidefinite programming. The
  first method, \emph{maximum mean absolute deviation rounding}
  (\mdr), seeks directions of large spread in the data while damping
  the effect of outliers.  The second method produces a
  \emph{low-leverage decomposition} (\lld) of the data that attempts
  to form a low-rank model for the data by separating out corrupted
  observations.  This paper also presents efficient computational
  methods for solving these SDPs.  Numerical experiments confirm the
  value of these new techniques.
\end{abstract}

\maketitle

\section{Introduction}
\label{sec:introduction}

Principal component analysis (PCA), proposed in 1933 by
Hotelling~\cite{Hotelling1933}, is a common technique for summarizing
high-dimensional data.  Principal components are designed to identify
directions in which the observations vary most.  As a consequence, PCA
is often used to reduce the dimension of the data.

Statistics based on variance, such as principal components, are highly
sensitive to outliers~\cite{Tukey1960}.  The literature on robust
statistics contains a wide variety of techniques that attempt to
correct this shortcoming~\cite{Huber2009}.  Unfortunately, many of
these approaches are based on intractable optimization problems or
lack a principled foundation.

Our focus in this work is to develop new formulations for robust PCA
that can be solved efficiently using convex programming algorithms.
Our first proposal, which we call \emph{maximum mean absolute
  deviation rounding} (\mdr), exchanges the variance in the definition
of PCA with a function less sensitive to outliers known as the mean
absolute deviation.  Although this formulation leads to a non-convex
optimization problem, we demonstrate that it is possible to
approximate the optimum by relaxing to a semidefinite program and
randomly rounding the solution.  This method can be viewed as a
specific instance of projection-pursuit PCA~\cite{Li1985}.

Our second proposal uses a different semidefinite program to split the
input data into the sum of a low-leverage matrix and a matrix of
corrupted observations.  We refer to this dissection as a
\emph{low-leverage decomposition} (\lld) of the data.  This method is
similar in spirit to the rank-sparsity decomposition of Chandrasekaran
et~al.~\cite{Chandrasekaran2009}.  While preparing this manuscript, we
learned of an independent investigation into this formulation of
robust PCA by Xu et.~al.\cite{Xu2010,Xu2010a}.

We describe algorithms that solve these semidefinite programs
efficiently, and we provide numerical experiments that confirm the
effectiveness of these new techniques.  We begin with a brief overview
of our proposals before laying out the details in
Sections~\ref{sec:mdr} and~\ref{sec:lld}.

\subsection{The Data Model}
\label{sec:datamodel}
Suppose that we have a family $\{\vec x_i\}_{i=1}^n$ of $n$
observations in $p$ dimensions.  We form an $n\times p$ data matrix
$\mat X$ whose rows are the observations.  The observations are
assumed to be centered; that is, $\frac{1}{n}\sum_i \vec x_i \approx
\mathbf{0}$.  While our methods do not explicitly require the data to
be centered, this hypothesis allows us to interpret principal
components as directions of high variance in the data.  We discuss
practical centering approaches in Section~\ref{sec:experiments}.

\subsection{Maximizing the Mean Absolute Deviation}
\label{sec:submdr}

Our first method is designed to mitigate a source of sensitivity in
classical principal component analysis.  The top principal component
$\vec v_\pca$ is defined as a direction of maximum variance in the
data:
\begin{equation}
  \label{eq:pca}
  \vec v_\pca  = 
  \argmax_{\ltwo{\vec v } = 1}
  \sum\nolimits_{i=1}^n \abs{\Inner{ \vec x_i , \vec v}}^2.
\end{equation}
The squared inner products in~\eqref{eq:pca} may lead to outsized
influence of outlying points because squaring a large number results
in a huge number, which can drag the principal component away from the
bulk of the data.  We can reduce this effect by replacing the squared
inner product with a measure of spread that is less sensitive. We
propose the use of the absolute value of the inner product:
\begin{equation}
  \label{eq:mdpca}
 \vec v_{\mathrm{MD}} = \argmax_{\ltwo{\vec v}= 1}
  \sum\nolimits_{i=1}^n \abs{\Inner{\vec x_i, \vec v}},
\end{equation}
where we have added the subscript $\mathrm{MD}$ to indicate that we
have exchanged the variance in equation~\eqref{eq:pca} with a measure
of spread known as the \emph{mean absolute deviation}
(MD)~\cite[p.~2]{Huber2009}.

This revision results in some complications.  The
formulation~\eqref{eq:pca} is an eigenvector problem which can be
solved efficiently.  In contrast, it is \NP-hard to compute $\vec
v_\md$.  Nevertheless, we develop an efficient randomized algorithm
that provably computes an approximate solution to~\eqref{eq:mdpca}.
We call this approach \emph{maximum mean absolute deviation rounding}
(\mdr).

Our main result, Theorem~\ref{thm:mdr},
states that, for any failure probability $\delta > 0$ and loss factor
$\eps > 0$, our algorithm produces a unit-norm vector $\vec v_{\mdr}$
such that
\begin{equation*}
  \sum\nolimits_{i=1}^n  \abs{\Inner{\vec x_i, \vec v_\mdr}} \ge 
  \sqrt{\frac{2}{\pi}}\left(1 -\eps\right)
  \max_{\ltwo{\vec v } = 1} \sum\nolimits_{i=1}^n \abs{\Inner{\vec x_i, \vec v}}.
\end{equation*}
The algorithm requires that we solve one semidefinite program (SDP)
whose size is polynomial in the number of observations.  Since SDPs
are solvable in polynomial time using interior-point methods, our
algorithm is tractable in principle.  In practice, solving SDPs can be
daunting even for moderately sized input data---say, more than 100
observations.  To address this issue, we detail a technique of Burer
and Monteiro~\cite{Burer2003a,Burer2004a} that can usually solve the
SDP efficiently, and in Section~\ref{sec:experiments} we provide some
numerical evidence that this approach succeeds.

We find additional components by greedily restricting the data to a
subspace perpendicular to the previous components and
solving~\eqref{eq:mdpca} again. 

This proposal is not without precedent. A more general formulation appears in
Huber's book~\cite[p.  203]{Huber1981}, and it is now known as
\emph{projection-pursuit PCA} (PP-PCA)~\cite{Li1985}.  We provide
further detail on PP-PCA in Section~\ref{sec:pppca} and discuss the
history of the method in~\ref{sec:pppca_background}.

\subsection{A Low-Leverage Decomposition}
\label{sec:lldintro}

Our second proposal stems from a different interpretation of classical
principal component analysis. Instead of viewing classical principal
components as directions of maximum variance, we can view them as an
optimal low-rank model for the data~\cite{Candes2009}. Suppose $\mat
P_\opt$ is a matrix that solves
\begin{equation*}
  \minprog{}{\fronorm{\mat X - \mat P}}{\rank(\mat P) = T.}
\end{equation*}
The dominant principal components of $\mat X$ are given by the $T$
right singular vectors of $\mat P_\opt$ corresponding with the nonzero
singular values of $\mat P_\opt$.

With real data, one is often faced with the situation where entire
observations are corrupted.  If this is the case, we would still like
to recover a low-rank model.  We can develop as natural formulation
for identifying a low-rank model using the well-known rank
sparsity~\cite{Fazel2002} and group sparsity~\cite{Rao1998a}
heuristics.  We propose to decompose the data matrix as $\mat X = \mat
P_\lld + \mat C_\lld$ by solving the semidefinite program
\begin{equation}
  \label{eq:lldintro}
  \minprog{}{\sum_{i} \sigma_i(\mat P) + \gamma \sum_j \ltwo{\vec c_j}}
 {\mat P + \mat C = \mat X.}
\end{equation}
We have written $\sigma_i(\mat P)$ for the $i$th singular value of
$\mat P$ and $\vec c_i$ for the $i$th row of $\mat C$.

We view the optimal matrix $\mat P_\lld$ as a surrogate for the
low-rank approximation to the uncorrupted data, and the optimal matrix
$\mat C_\lld$ as an approximation of the corrupted data.  The
formulation~\eqref{eq:lldintro} has an interesting property even when
$\mat P_\lld$ is not low-rank or $\mat C_\lld$ is not row-sparse:
$\mat P_\lld$ is guaranteed to be a low-leverage set of observations
in a sense we make precise in Section~\ref{sec:justlld}. As a result,
we refer to $\mat X = \mat P_\lld + \mat C_\lld$ as a
\emph{low-leverage decomposition} (LLD) of the data.  We define the
dominant \lld\ components as the right singular vectors of $\mat
P_\lld$.

This optimization problem is similar to the rank-sparsity
decomposition problem proposed in~\cite{Chandrasekaran2009}; see
also~\cite{Candes2009}.  We discuss these ideas at more length in
Section~\ref{sec:prevwork}.  As this manuscript was being prepared, we
learned of an independent investigation of the
program~\eqref{eq:lldintro} for robust PCA by Xu
et.~al.~\cite{Xu2010,Xu2010a} that provides conditions for recovery of
the support of the corruption and the row-space of the uncorrupted
observations.

\subsection{Road map}
\label{sec:roadmap}
Sections~\ref{sec:mdr} and~\ref{sec:lld} describe our proposals in
more detail, including theoretical guarantees and practical
algorithms.  Section~\ref{sec:prevwork} offers an overview of previous
work on robust PCA, while Section~\ref{sec:experiments} describes
numerical experiments illustrating the performance of our methods in
various settings.  A technical appendix contains the proofs of
supporting results.

\subsection{Notation}
\label{sec:notation}
We work exclusively with real numbers.  The symbols $\prob$ and
$\xpcd$ denote probability and expectation, respectively. We use
$\subg$ to denote the subgradient map.

Bold capital letters denote matrices while bold lower-case letters
denote vectors.  We represent the $i$th row of a matrix $\mat A$ by
$\vec a_i$ and the $j$th entry of a vector $\vec a$ by $a_j$.  The
adjoint of a matrix $\mat A$ is written $\mat A^\ad$.  When referring
to matrix elements, we sometimes use the notation $[\mat A]_{ij}$, and
similarly for vectors we use $[\vec a]_i$.

We use the compact convention for the singular value decomposition
(SVD) of a matrix: when $\mat A$ is rank $r$, we write its SVD as
$\mat A = \mat U \mat \Sigma \mat V^{\ad}$, where $\mat U$ and $\mat
V$ have orthonormal columns, and $\mat \Sigma$ is a non-singular
diagonal matrix whose entries are positive and are arranged in weakly
decreasing order.  The notation $\mat A\ggeq\mat B$ denotes that $\mat
A - \mat B$ is positive semidefinite.

\subsubsection{Norms}
\label{sec:norms}

We denote the $\ell_p$ vector norm as $\norm{\vec u}_p= \left(\sum_i
  \abs{u_i}^p\right)^{1/p}$ for $1 \le p < \infty$ and $\norm{\vec
  u}_\infty = \max_i \abs{u_i}$.  The Frobenius norm of a matrix is
defined by $\fronorm{\mat A}^2 = \Inner{\mat A, \mat A} $, where
$\Inner{\cdot,\cdot}$ represents the standard inner product.  The
Moore--Penrose pseudoinverse of a matrix $\mat A$ is denoted $\mat
A^\dagger$.

We define the $\ell_p$ to $\ell_q$ operator norm and its dual
respectively by
\begin{equation*}
  \norm{\mat A}_{p \to q} = \sup_{\norm{\vec u}_p = 1} \norm{\mat A
  \vec u}_q, \quad \text{ and } \quad \norm{\mat B}_{p \to q}^* = \sup_{\norm{\mat A}_{p
  \to q} = 1} \Inner{\mat B, \mat A}.
\end{equation*}
Table~\ref{tab:normsummary} describes some of the specific operator
norms used in this work.  We  also use the norms $\norm{\mat A}_{2
  \to 1}$ and $\norm{\mat A}_{\infty \to 1}$, which lack such simple
descriptions; see Sections~\ref{sec:mdpca_is_hard}
and~\ref{sec:rounding}.  

The operator norm of the adjoint satisfies $\norm{\mat A^\ad}_{q^* \to
  p^*} = \norm{\mat A}_{p \to q} $ where $p$ and $q$ satisfy the
conjugacy relations $1/p + 1/p^* = 1$ and $1/q + 1/q^* = 1$ with the
convention $1/\infty =0$.

\renewcommand{\arraystretch}{1.3}
\begin{table}[htbp]
  \centering
  \caption{Summary of the norms used in this work.}
  \begin{tabular}{|c|c|c|}
    \hline
    Norm & Description & Description of Dual \\
    \hline
    $\norm{\mat A}_{2\to 2}$    & Maximum singular value of $\mat A$ & Sum of the singular values of $\mat A$ \\
    \hline
    $\norm{\mat A}_{2 \to \infty}$ & Maximum $\ell_2$ row norm of $\mat A$ & Sum of the $\ell_2 $ row norms  of $\mat A$\\
    \hline
    $\norm{\mat A}_{1 \to \infty}$    & Maximum absolute entry of $\mat A$ & Sum of the absolute entries of $\mat A$ \\
    \hline
  \end{tabular}
  \label{tab:normsummary}
\end{table}

\section{Maximum Mean Absolute Deviation Rounding}
\label{sec:mdr}

Our first method is based on the classical interpretation of the top
principal component as the direction of maximum empirical variance in
multidimensional data.  It has long been recognized that the variance
is highly sensitive to outliers in the data~\cite{Tukey1960}.  The
field of robust statistics has reacted by developing and analyzing
robust measures of spread known as robust scales; see \cite[Ch.
5]{Huber2009} or~\cite[Sec. 2.5]{Maronna2006}.  This literature
describes a generic method for determining robust principal components
by replacing the variance with a robust measure of scale.  Li and
Chen~\cite{Li1985} published the first investigation of this under the
name \emph{projection-pursuit PCA} (PP-PCA).  Our proposal is a
specific instance of PP-PCA with the mean absolute deviation
scale~\eqref{eq:md}.  We show that this formulation is computationally
intractable, but we develop an algorithm that provably approximates
its solution.  To our knowledge, this is the first rigorous algorithm
for PP-PCA with a robust scale.

\subsection{Scales}
\label{sec:robscales}

A \emph{scale} is a function that measures the spread of
one-dimensional data~\cite[Ch. 5]{Huber2009}.  By far, the most common
scale is the empirical standard deviation, defined\footnote{One
  usually defines scales so that they are unbiased estimates of the
  sample standard deviation when the data is drawn from a normal
  distribution.  We are more interested in the direction of maximal
  scale rather than the value, so we can safely ignore the
  normalization factor.} as
\begin{equation*}
  \std(\vec y) = \left(\sum\nolimits_i y_i^2\right)^{1/2}  = \ltwo{\vec y},
\end{equation*}
where we we assume the data $\vec y$ is centered.  Of course, the
standard deviation is not the only way to measure the spread of the
data.  An alternative proposal~\cite[p. 2]{Huber2009} is the
\emph{mean absolute deviation} (MD).  For centered data $\vec y$, the
MD scale is defined as
\begin{equation}\label{eq:md}
  \md(\vec y)  = \sum_i \abs{y_i} = \lone{\vec y}.
\end{equation}
More generally, a scale is a function $S:\Re^n \to \Re$ such that
$S(\alpha \vec y) =\abs{\alpha} S(\vec y)$.  Scales are typically
chosen so that they are less sensitive to outliers than the standard
deviation.  The robust statistics literature focuses on scales that
have a positive breakdown point: the value of the scale cannot be
arbitrarily corrupted by nefariously chosen observations, so long as
the fraction of bad observations in the entire data set is small.
Although the mean absolute deviation has a breakdown point of zero, it
exhibits more efficient behavior than the standard deviation under
contaminated distributions~\cite{Tukey1960}.

\subsubsection{Scales for multivariate data}
\label{sec:highscale}

We extend the definition of scales to multivariate data by considering
the scale of the data in a given direction.  The projection of the
rows of $\mat X$ onto the unit direction $\vec u$ is given by the
product $\mat X \vec u$.  Note that if $\mat X$ is centered in the
sense of Section~\ref{sec:datamodel}, then the projection $\mat X \mat
u$ is also centered by linearity.  We define the scale of $\mat X$ in
the direction $\vec u$ to be the scale of the projected data $S(\mat X
\vec u)$.

As noted in~\cite{Huber1981}, this definition is equivariant under an
orthogonal change of basis: for any $\mat Q$ with $\mat Q^\ad \mat Q =
\mathbf{I}$, the scale of $\mat X$ in the direction $\vec u$ is
equal to the scale of $\mat X \mat Q^\ad$ in the direction $\mat Q
\vec u$.

\subsection{Projection-Pursuit PCA}
\label{sec:pppca}

Classically, the top principal component is defined as the direction
where the empirical standard deviation in the data is largest:
\begin{equation}
  \label{eq:varianceinv}
  \vec v_\pca = \argmax_{\ltwo{\vec v} = 1} \; \std(\mat X \vec v).
\end{equation}
A natural approach for finding robust components is to replace the
standard deviation in~\eqref{eq:varianceinv} with a robust scale
$S(\cdot)$, so that the robust component is the direction of maximum
robust scale
\begin{equation*}
\vec v_\pp = \argmax_{\ltwo{\vec v} = 1} S(\mat X \vec v).  
\end{equation*}
We define further robust components inductively by adding
orthogonality constraints:
\begin{equation}
  \label{eq:orthRestPPPCA}
  \vec v_\pp^{(k)} = \argmax_{\substack{\ltwo{\vec v} = 1 \\ 
      \vec v \perp \vec v_\pp^{(j)} \; \forall\, j < k}} 
  S(\mat X \vec v).
\end{equation}
This greedy method of constructing orthogonal components based on
robust scales goes by the name projection-pursuit PCA.  This scheme
was originally proposed by Huber~\cite[p. 203]{Huber1981}, but was
first studied in detail by Li and Chen~\cite{Li1985}.  PP-PCA reduces
to PCA when the scale is given by the standard deviation due to the
variational characterization of eigenvectors by Courant and Fischer.

To implement the PP-PCA method, one only needs a method that finds the
first component.  We discuss how to enforce the orthogonality
constraints in Section~\ref{sec:orthrest}.

\subsection{PP-PCA with the MD Scale is \textsf{NP}-Hard}
\label{sec:mdpca_is_hard}

Finding the top principal component is an eigenvector problem that
amounts to computing the direction where the norm $\norm{\cdot}_{2\to
  2}$ is achieved.  Similarly, PP-PCA with the MD scale amounts to
finding a vector that achieves an operator norm.  Indeed, the problem
$\vec v_{\md} = \argmax_{\ltwo{\vec v}=1} \lone{\mat X \vec v}$ is
equivalent to the problem
\begin{equation}
  \label{eq:mdop}
  \text{find } \ltwo{\vec v_\md }=1 \text{ such that } 
 \lone{\mat X \vec v_\md} = \norm{\mat X}_{2\to 1}.
\end{equation}
Unfortunately, exchanging the $\ell_2$ norm for the $\ell_1$ norm
leads to an \NP-hard computational problem.  To see this, we require
the following result, which we establish in the Appendix.
\begin{fact}\label{prop:inf1factor}
  For each matrix $\mat X$, the identity $\norm{\mat X}_{2\to 1}^2 =
  \norm{\mat X \mat X^\ad}_{\infty\to 1}$ holds.
\end{fact} 
 
Rohn~\cite{Rohn2000a} shows that there exists a class of
well-conditioned positive matrices $\mathcal{M}$ such that the
existence of a polynomial-time algorithm for accurately computing
$\norm{\mat M}_{\infty \to 1}$ for all $\mat M \in \mathcal{M}$
implies $\mathsf{P}=\mathsf{NP}$.  Since we can factor positive
matrices $\mat M = \mat R \mat R^\ad$ in polynomial time using, for example,
a Cholesky factorization, the existence of an accurate polynomial-time
algorithm that computes $\norm{\mat R}_{2\to 1}^2$ for any matrix
$\mat R$ implies that $\mathsf{P}=\mathsf{NP}$.

The observation that Equation~\eqref{eq:orthRestPPPCA} is \NP-hard to
solve for the specific choice $S(\cdot) = \lone{\cdot}$ has serious
implications for existing PP-PCA algorithms.  The algorithms available
in the literature for PP-PCA~\cite{Croux2007,Croux2005,Li1985} are
general schemes that claim to work for any choice of scale $S$.  As
a result, none of these algorithms can provide both accurate and
efficient solutions to the PP-PCA problem.  This issue is not merely
theoretical because these algorithms tend to perform poorly in
practice.  We discuss this point further in
Section~\ref{sec:pppca_background}.

\subsection{Approximating the $\ell_2\to \ell_1$ Norm using Randomized
  Rounding}\label{sec:rounding}

Although it is \NP-hard to compute the $\ell_2\to\ell_1$ norm, it is
possible to approximate its value efficiently.  This fact is a
consequence of the little Grothendieck
theorem~\cite[Sec. 5b]{Pisier1986a}, but the algorithm depends on
ideas of Nesterov~\cite{Nesterov1998}, a technique of Burer and
Monteiro~\cite{Burer2003a,Burer2004a}, and a new factorization step.

\subsubsection{The semidefinite relaxation of the $\ell_2 \to \ell_1$ norm}
\label{sec:sdprelax}

Before describing our algorithm, we begin by showing how the
computation of $\ell_2\to \ell_1$ operator norm can be relaxed to a
semidefinite program.  First, apply Fact~\ref{prop:inf1factor} to
change the computation of the $\ell_2 \to \ell_1$ norm to the
computation of the $\ell_\infty\to\ell_1$ norm:
\begin{equation}
  \label{eq:2to1->1toinf}
  \norm{\mat X}_{2\to 1}^2 = \norm{\mat X\mat X^\ad}_{\infty \to 1}
  = \max_{\norm{\vec y}_{\infty}  = 1}
  \vec y^\ad \mat X \mat X^\ad \vec y.
\end{equation}
The second identity above follows from the proof of
Fact~\ref{prop:inf1factor}; see also~\cite[Prop. 1]{Rohn2000a}.
Interpreting the quadratic form on the right hand side
of~\eqref{eq:2to1->1toinf} as a trace implies that $\norm{\mat
  X}_{2\to 1}^2$ is the optimal value of the (non-convex) program
\begin{equation}
  \label{eq:relax}
  \maxprog{}{ \trace(\mat X \mat X^\ad \mat Z) }{ \mat Z = \vec y \vec y^\ad, 
    \quad [\mat Z]_{ii} = 1  \text{ for all } i. }
\end{equation}
Relaxing the rank one constraint $\mat Z = \vec y \vec y^\ad$ to a
positive-semidefinite constraint $\mat Z \ggeq \mathbf{0}$ leads to
the SDP
\begin{equation}
  \label{eq:relax2}
    \maxprog{}{\trace(\mat X \mat X^\ad \mat Z)}
    {\mat Z \ggeq \mathbf{0},\quad [\mat Z]_{ii} = 1 
        \text{ for all } i.}
\end{equation}
It follows that $\norm{\mat X}_{2\to 1} \le \alpha_\opt$, where
$\alpha_\opt^2$ is the optimal value of~\eqref{eq:relax2}.  Moreover,
Grothendieck's inequality for positive-semidefinite matrices implies that
\begin{equation}
  \label{eq:alphabdd}
  \alpha_\opt^2 \le \frac{\pi}{2} \norm{\mat X \mat X^\ad}_{\infty \to 1},
\end{equation}
where this inequality is asymptotically the best
possible~\cite[Sec. 4.2]{Alon2006}.  Thus, $\alpha_\opt$ is within a
factor of $\sqrt{\pi/2}<1.26$ of the true value of the norm
$\norm{\mat X}_{2\to 1}$.

\subsection{The MDR Algorithm}
\label{sec:mdralg}

The fact that equation~\eqref{eq:relax2} gives us a good upper bound
on the \emph{value} of $\norm{\mat X}_{2\to 1}$ is of secondary
importance.  We would prefer an approximation for $\vec v_\md$
in~\eqref{eq:mdop}, that is, a vector $\vec v_\opt$ with $\ltwo{\vec
  v_\opt} = 1$ such that $\norm{\mat X \vec v_\opt} \approx \norm{\mat
  X}_{2\to 1}$.  We accomplish this goal via a randomized procedure
that rounds an optimal solution $\mat Z_\opt$ to~\eqref{eq:relax2}
back to a vector $\vec v_\opt$.  The entire procedure is detailed in
Algorithm~\ref{alg:mdr}.

The first step of the algorithm solves the SDP
relaxation~\eqref{eq:relax2}. In Step~\ref{step:mdr_random}(a), we
draw a random $\vec y \in \{\pm 1\}^n$ with $\xpcd \lone{\mat X \mat
  X^\ad \vec y } = 2\alpha_\opt^2/\pi $. This procedure is well
understood~\cite{Nesterov1998}.  The method in
Step~\ref{step:mdr_random}(b) that we use to compute $\vec v$ from
$\vec y$ is novel, and it requires a proof of correctness, which
appears in the Appendix.  By choosing the best random outcome,
Step~\ref{step:mdr_max} limits the probability that our method fails
to provide a reasonable approximation.

The following theorem describes the behavior of
Algorithm~\ref{alg:mdr}.
\begin{theorem}
  \label{thm:mdr}
  Suppose that $\mat X$ is an $n\times p$ matrix, and let $K$ be the
  number of rounding trials.  Let $(\vec v_\opt, \alpha_\opt)$ be the
  output of Algorithm~\ref{alg:mdr}.  Then $\alpha_\opt \ge \norm{\mat
    X}_{2\to 1}$. Moreover, for $\theta < 1$, the inequality
  \begin{equation}\label{eq:mdrineq}
    \lone{\mat X \vec v_\opt } > \theta \sqrt{\frac{2}{\pi}}
    \alpha_\opt    
  \end{equation}
  holds except with probability $\mathrm{e}^{-2K(1-\theta^2)/\pi}$.
\end{theorem}

In Theorem~\ref{thm:mdr}, it may be more natural to specify a failure
probability $\delta >0$ and approximation loss $\eps = 1-\theta>0$
instead of a repetition number $K$.  In this case, simple algebra
shows that $\lone{\mat X \vec v_\opt} > (1-\eps)\sqrt{2/\pi}
\norm{\mat X}_{2\to 1}$ except with probability $\delta$, so long as
\[K \ge \frac{\pi}{2}\cdot \frac{\log(1/\delta)}{\eps(2-\eps)} =
\Oh\left(\eps^{-1}\log(\delta^{-1})\right).\] In particular, the
choice $K = 94$ implies that $\lone{\mat X \vec v_\opt} > 0.75
\norm{\mat X}_{2\to 1}$ with probability at least $0.999$.

We use the approximation ratio $\rho = \lone{\mat X \vec
  v_\opt}/\alpha_\opt $ to measure the quality of the optimal solution
in Section~\ref{sec:experiments}.  Although Theorem~\ref{thm:mdr} only
guarantees that we can make $\rho$ as close to $\sqrt{2/\pi} > 0.79$
as we desire, in practice we typically see a $0.95$ approximation
ratio or higher.  This observation does not indicate that the analysis
of the algorithm is loose; it follows directly from~\cite[Sec.
4.2]{Alon2006} that this bound is asymptotically tight for a class of
examples as $n\to \infty$.

\begin{figure}[t]
  \framebox{
    \begin{minipage}{0.9\columnwidth}
      \refstepcounter{algcounter}%
      \noindent{\large \textbf{Algorithm~\arabic{algcounter}:}} 
      \textbf{Maximum Mean Absolute Deviation Rounding} \label{alg:mdr}\\
      \textsc{Input:} An $n\times p$ matrix $\mat X$;
      repetition count $K$.\\
      \textsc{Output:} A $p\times 1$ unit-norm vector $\vec v_\opt$
      and an optimal value $\alpha_\opt$.  \vspace{5pt}
      \begin{enumerate}
        \setlength{\itemsep}{0pt}
        \setlength{\parskip}{0pt}
        \setlength{\parsep}{0pt}        
      \item  Find an $\mat R_\opt$ such that $\mat Z_\opt=\mat R_\opt
        \mat R_\opt^{\ad}$ solves the semidefinite program
        \begin{equation}\label{eq:mdrsdp}
          \begin{array}{ll} \maximize & \trace(\mat Z \mat X \mat X^{\ad}) \\
            \text{subject to} & \mat Z \ggeq \mathbf{0}, \quad [\mat Z]_{ii} = 1 
            \text{ for } i = 1,\dotsc, n
          \end{array}
        \end{equation}
        Set $\alpha_\opt$ to be the square root of the optimal value:
        $\alpha_\opt=\sqrt{\trace(\mat Z_\opt \mat X \mat X^\ad)}$. \label{alg:sdp}
      \item For each $k = 1,\dotsc,K$, do \label{step:mdr_random}
        \begin{enumerate}[(a)]
        \item Set $\vec y^{(k)} = \sign( \mat R_\opt \vec g^{(k)})$,
          where $\vec g^{(k)}$ is an $n\times 1$ standard normal
          random vector.
        \item Set ${\vec v^{(k)}} = \mat X^{\ad} \vec y^{(k)} /
          \ltwo{\mat X^{\ad}\vec y^{(k)}}$.
        \end{enumerate}
      \item \label{step:mdr_max} Set $\vec v_\opt = \argmax_{k = 1,\dotsc, K}
        \lone{\mat X \vec v^{(k)}}$.
      \end{enumerate}
    \end{minipage}}
\end{figure}

\subsection{Implementation of Algorithm~\ref{alg:mdr}}
\label{sec:mdr_imp}

For a fixed iteration count $K$, the complexity of
Algorithm~\ref{alg:mdr} is typically dominated by Step~\ref{alg:sdp}.
When applied to~\eqref{eq:mdrsdp}, modern interior-point methods are
guaranteed to compute the optimal objective value $\alpha_\opt$ and
optimal point $\mat Z_\opt$ accurately in polynomial time.  The factor
$\mat R_\opt$ is determined using a Cholesky factorization of $\mat
Z_\opt$.  In practice, interior-point methods are very slow for
large-scale problems, so we prefer an algorithm of Burer and
Monteiro~\cite{Burer2004a}.
 
The algorithm of Burer and Montiero never forms the semidefinite
matrix $\mat Z$; rather it operates directly with the factor $\mat R$.
We express the objective function of~\eqref{eq:mdrsdp} in terms of
$\mat R$ as $\trace( \mat R \mat R^{\ad} \mat X \mat X^{\ad}) =
\fronorm{\mat X^{\ad} \mat R}^2$.  The constraints $[\mat Z]_{ii} = 1$
are equivalent to constraints on the rows of $\mat R$ of the form
$\ltwo{\vec r_i} = 1$.

We implicitly enforce these row constraints by incorporating them into
the objective function as in~\cite[Sec. 4.2]{Burer2003a}.  The
resulting unconstrained, nonconvex optimization problem takes the form
\begin{equation}
  \label{eq:bm_opt}
 \maximize_{\mat R} \fronorm{\mat X^{\ad} \mathcal{N}(\mat R)}^2,
\end{equation}
where $\mathcal{N}(\mat R)$ denotes the operator that normalizes the
rows of $\mat R$, that is, $[\mathcal{N}(\mat R)]_{ij} = [\vec r_i]_j
/\ltwo{\vec r_i}$.

We then apply a conjugate gradient algorithm to maximize the
unconstrained objective in~\eqref{eq:bm_opt}.  Our particular
implementation uses the algorithm of Hager and Zhang~\cite{Hager2006},
which we have found to work well in our experiments.  We refer to our
online code for the choice of parameters in this conjugate gradient
algorithm~\cite{McCoy2010b}.

This factorization technique for solving~\eqref{eq:mdrsdp} is
advantageous because it reduces the dimension of the problem.  The
paper~\cite{Burer2004a} shows that restricting $\mat R$ to be an
$n\times k$ matrix for $k = \Oh(\sqrt{n})$ suffices to solve this
problem exactly.  To be precise, when $k = \floor{(1+\sqrt{9 + 8
    n})/2}$ any \emph{local} minimum $\mat R_\opt \in \Re^{n \times
  k}$ of~\eqref{eq:bm_opt} gives a \emph{global} minimum $\mat
Z_\opt$ of~\eqref{eq:mdrsdp} via the map $\mat Z_\opt = \mat R_\opt
\mat R_\opt^{\ad}$, provided a mild technical condition\footnotemark\ 
holds.

\footnotetext{Specifically, the objective function $\trace(\mat Z \mat
  X\mat X^{\ad})$ must not be constant along a face of the feasible
  set.}

\subsubsection{Orthogonal Restriction}\label{sec:orthrest}
Algorithm~\ref{alg:mdr} only approximates the first principal
component in~\eqref{eq:mdpca}.  In order to approximate the $k$th
robust principal component for $k>1$, we define a new matrix $\mat
X_k$ by restricting the rows of $\mat X$ to the subspace
perpendicular to the span of $\vec v_1,\dotsc,\vec v_{k-1}$.  Ignoring
numerical stability, we can inductively define
\begin{equation}\label{eq:naiveorth}
  \mat X_k = \mat X_{k-1} - \mat X \vec v_{k-1} \vec v_{k-1}^{\ad} 
  = \mat X \Bigl( \mathbf{I}- \sum\nolimits_{j=1}^{k-1} \vec v_j \vec v_j^\ad\Bigr),
\end{equation}
which ensures each row of $\mat X$ is orthogonal to the previous
components $\vec v_{j}$ for $j<k$.  We then apply
Algorithm~\ref{alg:mdr} to the restricted matrix $\mat X_k$ to produce
the component $\mat v_k$.  Since the output $\vec v_\opt$ of
Algorithm~\ref{alg:mdr} is a linear combination of the rows of the
input matrix by Step~\ref{step:mdr_random}(b), this iterative
procedure ensures that $\mat v_k$ is perpendicular to the previous
components.

In practice, the implementation can be done using Householder
reflections as in~\cite{Croux2005}; see~\cite{Stoer2002} for further
background on the implementation of Householder transformations.
Householder reflections are more numerically stable than the na\"ive
method~\eqref{eq:naiveorth}.  Moreover, they take full advantage of
the fact that we are only searching over a $p-k+1$ dimensional
subspace by reducing the dimension of $\mat X_{k}$ to
$n\times(p-k+1)$.

\subsection{Extending the Rounding to Multiple Components}
\label{sec:extensions}
We have also attempted to extract a collection of robust components
simultaneously by solving a single semidefinite program.  That is, we
would like to solve the problem
\begin{equation}
  \label{eq:allatonce}
  \maxprog{}{\sum_{i=1}^T \norm{\mat X \vec v_i}_1 }
  { \Inner{\vec v_i,\vec v_j } = \delta_{ij}},
\end{equation}
where $\delta_{ij}$ is the Kronecker delta function. When $T=1$,
equation~\eqref{eq:allatonce} is equivalent with~\eqref{eq:mdpca}.
When $T>1$, the restriction $\Inner{\vec v_i,\vec v_j} = \delta_{ij}$
ensures that the optimum occurs at an orthogonal set of unit vectors.

We can rephrase this optimization problem by the equivalent
quadratically constrained quadratic program
\begin{equation}
  \label{eq:qcqp}
  \maxprog{}{\sum_{i=1}^n \vec w_i^{\ad} \mat X \vec v_i}
  {
    \diag(\vec w_i \vec w_i^{\ad} ) = 1, \quad
    \Inner{\vec v_i,\vec v_j } = \delta_{ij}}
\end{equation}
The diagonal restrictions on $\vec w_i$ ensure that $\vec w_i \in
\{\pm 1\}^n$ for each $i = 1,\dotsc, n$.  The nonconvex
problem~\eqref{eq:qcqp} can be approximated via a semidefinite
relaxation proposed in~\cite{Nemirovski2007}.  The results
of~\cite{So2009} imply that the optimal value of this relaxation is
guaranteed to be larger than the optimal value of~\eqref{eq:allatonce}
by no more than a logarithmic factor.  The rounding procedure does not
produce orthogonal vectors, so we need to apply an additional
orthogonalization step to achieve feasibility
for~\eqref{eq:allatonce}. Empirically, we have found that the
orthogonalization increases the objective value over the standard
rounding, so it appears that there is no loss in applying this
procedure.

Unfortunately, this method does not appear to be competitive with the
projection pursuit method.  The vectors we find by coupling
Algorithm~\ref{alg:mdr} with the orthogonal pursuit of
Section~\ref{sec:orthrest} are feasible for~\eqref{eq:qcqp} and
typically provide a larger objective value than rounding coupled with
post-processing orthogonalization.  A better rounding procedure for
this type of relaxation may prove more effective than the
projection-pursuit approach; this is a direction for further research.

\section{The Low-Leverage Decomposition}
\label{sec:lld}

Our second method is derived from the interpretation of principal
component analysis as a matrix approximation problem.  When the
observations are drawn from a highly correlated family, the singular
values of the data matrix $\mat X$ tend to decay rapidly.  If this is
the case, then the matrix $\mat X$ is well approximated by a low-rank
matrix $\mat P$.

It is rare that a large data set can be compiled without error, but it
is often the case that the errors only affect a subset of the
observations.  We can model these errors through a multi-population
model.  Suppose that the bulk of the observations is well-explained by
a low-rank model while the remainder come from another population or
are corrupted by measurement noise.  A prudent approach to robust
principal component analysis would first separate the corrupted data
from the uncorrupted data before attempting to recover a low-rank
model.  When the corrupted rows are unknown, this task may seem
daunting.

To accomplish this task, we propose a semidefinite program that
decomposes the input $\mat X$ into two matrices:
\begin{equation}\label{eqn:lld}
  \minprog{(\mat P,\mat C)}{\nucnorm{\mat P}+\gamma \rsnorm{\mat C}}{
  \mat P+ \mat C=\mat X.}
\end{equation}
The norm $\nucnorm{\mat P}$ is the sum of the singular values of $\mat
P$ and is known to promote low-rank solutions~\cite{Fazel2002}, while
$\rsnorm{\mat C}$ is the sum of the $\ell_2$ norms of the rows of
$\mat C$ and promotes group sparsity~\cite{Rao1998a}.

We call the optimal matrix pair $(\mat P_\opt,\mat C_\opt)$ for the
problem~\eqref{eqn:lld} the \emph{low-leverage decomposition} (LLD) of
$\mat X$; we can interpret $\mat C_\opt$ as an identified corruption
and $\mat P_\opt$ as a surrogate for the uncorrupted observations.  We
define our robust components as the right singular vectors of the
surrogate matrix $\mat P_\opt$.  The detailed procedure appears in
Algorithm~\ref{alg:lld}.  We show in Section~\ref{sec:justlld} that
our recovered data matrix $\mat P_\opt$ has the additional property of
being a low-leverage set of observations.

The \lld\ formulation is related to recent
proposals~\cite{Candes2009,Chandrasekaran2009}, and we discuss this
point more in Section~\ref{sec:cvx_approach}.

As we were preparing this manuscript, we became aware of the
independent work~\cite{Xu2010,Xu2010a} which also considers~\eqref{eqn:lld}
for the robust PCA problem.  This work shows that, under certain
hypotheses, the recovered low-rank data $\mat P_\opt$ has the same
row-space as the true data and the corrupted rows are correctly
identified.  

\begin{figure}[t]
  \framebox{
    \begin{minipage}{0.8\columnwidth}
      \refstepcounter{algcounter}%
      \noindent{\large \textbf{Algorithm~\arabic{algcounter}:}} 
      \textbf{Low-Leverage Decomposition} \label{alg:lld} \\
      \textsc{Input:} An $n\times p$ data matrix $\mat X$;
      desired number of principal components $T$. \\
      \textsc{Output:} A $p\times T$ matrix $\mat V_\opt$ with
      orthogonal columns.
      \vspace{5pt}
      \begin{enumerate}
        \setlength{\itemsep}{0pt}
        \setlength{\parskip}{0pt}
        \setlength{\parsep}{0pt}        
      \item \label{Alg:LLD} Find $(\mat P_\opt,\mat C_\opt)$ that solve
        \begin{equation}\label{eq:LLDsdp}
          \minprog{(\mat P, \mat C)}{\nucnorm{\mat
              P}+\gamma\rsnorm{\mat C}}{\mat P+\mat C=\mat X}
        \end{equation}
      \item\label{Alg:eigs} Compute the SVD $\mat P_\opt = \mat U \mat
        \Sigma \mat V^{\ad}$.
      \item Set $\mat V_\opt$ to the first $T$ columns of $\mat V$, that is, set
        \begin{equation*}
          [\mat V_\opt]_{ij} = [\mat V]_{ij} \text{ for } i=1,\dotsc, p,\;\text{and}\; j=1,\dotsc, T.
        \end{equation*}
      \end{enumerate}
    \end{minipage}
  }
\end{figure}

\subsection{Low-Leverage by Duality}\label{sec:justlld}
In this section, we demonstrate that~\eqref{eqn:lld} extracts a
low-leverage model for the data.  This result follows from duality
arguments that characterize the optimum of the convex program.
\begin{lemma}[First-order optimality conditions
  for~\eqref{eqn:lld}]\label{lem:lldsubg}
  A feasible pair $(\mat P,\mat C)$ is optimal for~\eqref{eqn:lld} if
  and only if there exists a matrix $\mat Q$ such that
  \begin{subequations}\label{eqs:dualcert}
  \begin{align}
    \phantom{-}\Inner{\mat Q,\mat P} &= \phantom{\gamma}\nucnorm{\mat P}, 
    \quad
    \norm{\mat Q}_{2\to 2} \le 1 \label{Qspec}\\
    -\Inner{\mat Q,\mat C} &= \gamma \rsnorm{\mat C},
    \quad \rmnorm{\mat Q}\le \gamma, \label{Qrow}
  \end{align}
  \end{subequations}
\end{lemma}
\begin{proof} 
  It follows from standard subdifferential conditions that a feasible
  point $(\mat P,\mat C)$ minimizes the functional in~\eqref{eqn:lld}
  if and only if zero is in the subgradient of $f(\mat P) =
  \nucnorm{\mat P}+\gamma\rsnorm{\mat X-\mat P}$.  By the additivity
  of subgradients~\cite[Thm. 23.8]{Rockafellar1970}, this condition
  holds if and only if there exists a matrix $\mat Q$ such that the
  subgradient conditions $\mat Q \in \subg\nucnorm{\mat P}$ and $-\mat
  Q \in \subg \gamma \rsnorm{\mat C}$ are in force.
  
  We show that these subgradient conditions are equivalent
  to~\eqref{eqs:dualcert}.  By definition of the subdifferential, $\mat Q
  \in \subg \norm{\mat P}_{2\to 2}^*$ if and only if for every
  perturbation $\mat \Delta$ the subgradient inequality
  \begin{equation}
    \label{eq:subIneq}
    \Inner{\mat Q, \mat \Delta} \le 
    \norm{\mat P+\mat \Delta}_{2\to 2}^* - \norm{\mat P}_{2\to 2}^*
  \end{equation}
  holds.  Suppose first that~\eqref{Qspec} holds. Then, for all $\mat
  \Delta$, we have
  \begin{equation*}
    \Inner{\mat Q, \mat \Delta} = 
    \Inner{\mat Q, \mat P + \mat \Delta } - \norm{\mat P}_{2\to 2}^* 
    \le \norm{\mat Q}_{2\to 2} \norm{\mat P + \mat \Delta}_{2\to 2}^*  
    - \norm{\mat P}_{2\to 2}^* ,
  \end{equation*}
  where the inequality follows by the definition of dual norms.  Since
  $\norm{\mat Q}\le 1$ by assumption, the subgradient
  inequality~\eqref{eq:subIneq} must hold.

  It remains to show that the subgradient
  inequality~\eqref{eq:subIneq} implies~\eqref{Qspec}.  Taking $\mat
  \Delta = \mat P$ in~\eqref{eq:subIneq} gives $\Inner{\mat Q, \mat P}
  \le \norm{\mat P}_{2\to 2}^*$, while $\mat \Delta = - \mat P$ gives
  the reverse inequality $\Inner{\mat Q, \mat P} \ge \norm{\mat
    P}_{2\to 2}^*$.  Therefore the subgradient
  inequality~\eqref{eq:subIneq} implies $\Inner{\mat Q, \mat P } =
  \norm{\mat P}_{2\to 2}^*$.
  
  On the other hand, suppose that $\mat \Delta \ne \mathbf{0}$
  satisfies $\Inner{\mat Q, \mat \Delta } = \norm{\mat Q}_{2\to 2}
  \norm{\mat \Delta}_{2\to 2}^*$; such a matrix $\mat \Delta$ must
  always exist in finite dimensions since suprema are attained in the
  trace definition of norms.  Then the subgradient
  inequality~\eqref{eq:subIneq} implies
  \begin{equation*}
    \norm{\mat Q}_{2\to 2} \norm{\mat \Delta}_{2\to 2}^* \le
    \norm{\mat P+\mat \Delta }_{2\to 2}^* -\norm{\mat P}_{2\to 2}^*
    \le
    \norm{\mat \Delta}
  \end{equation*}
  where the second inequality follows by the triangle inequality.
  Since $\mat \Delta\ne \mathbf{0}$, we have shown that the
  subgradient inequality implies $\norm{\mat Q}_{2\to 2} \le 1$. Hence
  $\mat Q \in \subg \norm{\mat P}_{2\to 2}^*$ is equivalent
  to~\eqref{Qspec}.  The equivalence between $-\mat Q \in \subg \gamma
  \norm{\mat C}_{2\to \infty}^*$ and relation~\eqref{Qrow} follows
  analogously.
\end{proof}

Before continuing, we introduce another fact concerning the
subgradient of unitarily invariant norms.  Let $\mat P = \mat U \mat
\Sigma \mat V^\ad$ be the compact SVD of $\mat P$.  It follows
from~\cite{Watson1992} that~\eqref{Qspec} implies $\mat Q = \mat U
\mat V^\ad + \mat W$, where, in particular, $\mat U \mat V^\ad \mat W
= \mathbf{0}$.

\subsubsection{Leverage scores}
\label{sec:levscores}

The \emph{leverage score} of the observation $\mat x_i$ corresponding
to the $i$th row of $\mat X$ is given by the number $[\mat H]_{ii}$,
where $\mat H=\mat X(\mat X^*\mat X)^\dagger \mat X^*$ is the
orthoprojector onto the column space of $\mat X$.  We refer to $\mat
H$ as the \emph{hat matrix} in accord with common statistical
practice.  A large leverage score tends to indicate that the
corresponding observation lies outside of the bulk of the data,
although it does not necessarily indicate that the point is
influential in linear regression. We refer
to~\cite[Ch. 6]{Montgomery2006} for further discussion of leverage
scores.

The following theorem shows that the leverage scores of our
decomposition are bounded above by $\gamma^2$, justifying the
terminology low-leverage decomposition for the solution of the
program~\eqref{eqn:lld}.
\begin{theorem}\label{thm:lld}
  Suppose $(\mat P_\opt,\mat C_\opt)$ is an optimal point of the
  program~\eqref{eqn:lld}.  Then the diagonal elements of the hat
  matrix $\mat H=\mat P_\opt(\mat P_\opt^\ad\mat P_\opt)^\dagger \mat
  P_\opt^\ad$ are bounded above by $\gamma^2$.
\end{theorem}
\begin{proof}
  From the characterization of the subgradient of unitarily invariant
  norms~\cite{Watson1992} discussed above, we know that $\mat Q = \mat
  U\mat V^{\ad} + \mat W$ with $\mat U \mat V^{\ad} \mat W^{\ad} =
  \mathbf{0}$. Thus,
  \[
  \mat Q\mat Q^{\ad} = 
  \mat U\mat U^{\ad} + \mat W\mat W^{\ad} \ggeq \mat U\mat
  U^{\ad} = \mat H,
  \] 
  where the last equality can be easily checked using the definition
  of $\mat H$ and the SVD of $\mat P_\opt$.  Since the diagonal
  entries of a positive-semidefinite matrix are nonnegative, this
  relation implies $[\mat H]_{ii} \le [\mat Q\mat Q^{\ad}]_{ii}$.
  Recall that the $\ell_2 \to \ell_\infty$ operator norm is the
  maximum $\ell_2$ row norm of the matrix. Thus relation~\eqref{Qrow}
  of Lemma~\ref{lem:lldsubg} implies that $[\mat Q\mat Q^{\ad}]_{ii}
  \le \gamma^2$, which completes the proof.
\end{proof}

We can view our proposal as a method of decomposing a data matrix
$\mat X$ into a component with a (user-specified!) upper bound on the
leverage plus an error term. Moreover, this result gives a statistical
interpretation to the regularization parameter $\gamma$
in~\eqref{eqn:lld}.

We note that while our program guarantees a low-leverage
decomposition, an assumption of suitably small leverage is a technical
hypothesis in other works, e.g.,~\cite[eq.~(1.2)]{Candes2009}.

The reader should be warned that this method does not necessarily
produce a low-leverage solution if we use our program to identify
outlying data and then ``prune'' the rows. That is, suppose $(\mat
P_\opt,\mat C_\opt)$ is an optimal point of~\eqref{eqn:lld} and $\vec
c_i = \mathbf{0}$ for row indices $i \in I$. Then the corresponding
matrix $\mat X_I= \mat P_I$ \emph{does not} necessarily have leverage
scores bounded above by $\gamma^2$.

\subsection{The Choice of $\gamma$ }
\label{sec:lld_gamma}

In this section, we study how the value of the regularization
parameter $\gamma$ affects the properties of the decomposition.

We begin by showing that, when $\gamma \ge 1$, the degenerate solution
$(\mat P_\opt,\mat C_\opt) = (\mat X, \mathbf{0})$
minimizes~\eqref{eqn:lld}.  This claim follows by explicit
construction.  Let $\mat U \mat \Sigma \mat V^{\ad}$ be the compact
SVD of $\mat X$, and define $\mat Q = \mat U \mat V^{\ad}$.  Clearly
$\Inner{\mat Q, \mat X}= \nucnorm{\mat X}$, so $\mat Q$
satisfies~\eqref{Qspec} with $\mat P_\opt = \mat X$.  By construction,
the maximum singular value of $\mat Q$ is bounded above by one.
Equivalently, $\mat Q \mat Q^\ad \gleq \mathbf{I}$. This inequality
implies $[\mat Q \mat Q^\ad]_{ii}\le 1$.  Since the diagonal entries
of $\mat Q \mat Q^\ad$ are the squared row norms of $\mat Q$, we have
shown that $\norm{\mat Q}_{2\to \infty} \le 1 \le \gamma$.  This bound
demonstrates that $\mat Q$ satisfies~\eqref{Qrow} with $\mat C_\opt =
\mathbf{0}$, which certifies optimality of this degenerate solution by
Lemma~\ref{lem:lldsubg}.

We now show that the regularization parameter $\gamma$ gives an upper
bound on the rank of the optimal $\mat P_\opt$.  It is easy to show
using the SVD of $\mat P_\opt$ that the trace of the hat matrix $\mat
H$ defined above is equal the rank of $\mat P_\opt$.  Since $[\mat
H]_{ii} \le \gamma^2$ by Theorem~\ref{thm:lld}, we must have
\begin{equation}\label{eq:gammarank}
  \rank(\mat P_\opt) = \trace(\mat H) \le n \gamma^2.
\end{equation}
The rank is a positive integer, so $\gamma < 1/\sqrt{n}$ implies that
the optimal $\mat P_\opt$ is trivial.  Moreover, in order to get $T$
meaningful components in Step~\ref{Alg:eigs} of
Algorithm~\ref{alg:lld}, we require $\rank(\mat P_\opt) \ge T$.  Thus,
we can limit ourselves to situations where $\gamma \in [\sqrt{T/n},
1]$.

Inequality~\eqref{eq:gammarank} has implications for the numerical
solution of~\eqref{eqn:lld}.  As we discuss in
Section~\ref{sec:lldalg}, the bulk of the computation comes from
computing an SVD at each iteration.  When the solution of the
optimization problem has low rank, the iterates also tend to have low
rank. This allows us to save significant computational effort by
computing partial singular decompositions at each step.  A judicious
choice of $\gamma$ can increase the performance of our algorithm
immensely.  We find that taking $n\gamma^2 \approx T^2$ is a useful
heuristic for achieving a rank-$T$ optimal solution, so long as $n \gg
T^2$.

On the other hand, typical statistical data does not show true
low-rank behavior even when there are no outliers.  Therefore, forcing
the optimal decomposition to be low rank typically results in a dense
corruption $\mat C_\opt$.  This effect may be mitigated somewhat by
another formulation we discuss briefly in Section~\ref{sec:lld_ext}.
In practice we find that setting $\gamma$ somewhat less than
$\sqrt{p/n}$, say $\gamma = 0.8\sqrt{p/n}$, provides a very good
low-rank model, but it does poorly in the context of outlier
identification.  We discuss specific parameter choices for our
experiments in Section~\ref{sec:experiments}.

\subsection{Computing the Low-Leverage Decomposition}
\label{sec:lldalg}
Although general-purpose semidefinite programming software such as
\texttt{CVX}~\cite{Grant2010,Grant2008} can solve small instances
of~\eqref{eqn:lld} efficiently, the interior-point methods they
utilize may be unable to complete even a single iteration of a
large-scale problem.  This observation indicates that we need to use
different methods for large-scale problems.

To solve~\eqref{eqn:lld}, we recommend an alternating direction
augmented Lagrangian algorithm analogous to the one used
in~\cite{Candes2009}; see also~\cite{Lin2009}. The generic form of the
method is known as the Augmented Lagrangian Method of Multipliers
(ALMM).  The augmented Lagrangian for~\eqref{eqn:lld} with dual
variable $\mat Q$ is given by
\begin{equation}
  \mathcal{L}_\mu(\mat P,\mat C,\mat Q) = \nucnorm{\mat P} + \gamma
  \rsnorm{\mat C} + \Inner{\mat X - \mat P - \mat C, \mat Q} +
  \frac{\mu}{2} \fronorm{\mat X - \mat P - \mat C}^2.
\end{equation}
For an initial starting point $\mat P^0$, we alternately solve $\mat
P^{k+1} = \argmin_{\mat P} \mathcal{L}_\mu(\mat P, \mat C^k, \mat
Q^k)$ and $\mat C^{k+1} =\argmin_{\mat C} \mathcal{L}_\mu(\mat
P^{k+1}, \mat C, \mat Q^k)$.  We then update the multiplier by the
feasibility gap $\mat Q^{k+1} = \mat Q^k + \mu (\mat X - \mat P^{k+1}
- \mat C^{k+1})$.

The minimizations above have an explicit form in terms of shrinkage
operations~\cite{Combettes2006}%
\begin{subequations}\label{eqs:shrinks}
  \begin{align} 
    \mat C^{k+1} &= \shrinkrows\left(\mat X - \mat P^k + \frac{1}{\mu}
      \mat Q^k,\mu \gamma \right) \\ 
    \mat P^{k+1} &= \shrinkspec\left(\mat
      X - \mat C^{k+1} + \frac{1}{\mu} \mat Q^k,\mu \right),
  \end{align}
\end{subequations}
where $\shrinkrows(\mat A,\nu)$ soft-thresholds each row $\vec a_i$ of $\mat
A$:
\begin{equation*}
  \shrinkrows(:,\nu): \mat A \longmapsto 
  \diag([1-\nu/\ltwo{\vec a_i}]_+)\cdot \mat A,
\end{equation*}
where $[x]_+ = \max\{x,0\}$.  Similarly $\shrinkspec(\mat A, \nu)$
soft-thresholds the singular values of $\mat A$
\begin{equation}
  \shrinkspec(:, \nu):  \mat U \mat \Sigma \mat V^\ad \longmapsto \mat U 
  \left[\mat \Sigma -\nu \mathbf{I}\right]_+
  \mat V^{\ad},
\end{equation}
where the operator $[\cdot]_+$ is applied element-wise.  We initialize
the algorithm with $\mat P^0 = \mathbf{0}$ and set the parameter $\mu
= np/\rsnorm{\mat X}$.  We stop the algorithm when the iterates are
nearly feasible, that is, $\bigl\|{\mat X - \mat P^k -\mat
  C^k}\bigr\|<10^{-7}\fronorm{\mat X}$.

The main computational difficulty when running this algorithm involves
computing the spectral shrinkage operator.  When the iterates $\mat
P^k$ are low rank, we can save significant computational effort by
performing only partial singular value decompositions~\cite{Lin2009}.
We can leverage our analysis in Section~\ref{sec:lld_gamma} to ensure
that the optimal $\mat P_\opt$ is low rank.  Since the algorithmic
iterates tend to be low-rank in this case, we can significantly
improve the performance of our algorithm by choosing $\gamma$ to limit
the rank of the optimal solution.  In practice, we have found that one
should set the quantity $n\gamma^2$ somewhat larger than the desired
rank of the solution, e.g., $n\gamma^2 \approx T^2$ when we desire a
rank-$T$ solution.

\subsection{Extensions for a Noisy Model}
\label{sec:lld_ext}

We note that there is an obvious extension of the \lld\ when one wants
to account for an additional of noise in the model.  Suppose
that in addition to gross corruptions of certain observations, we
would also like to model small corruptions or noise that may be spread
throughout the data.

Instead of enforcing the equality $\mat X = \mat P + \mat C$, we allow
for some additional slack of the form $\fronorm{\mat X - \mat P - \mat
  C} \le \eta$, where $\eta$ is an estimate for the noise level.  That
is, we solve the problem
\begin{equation}
  \label{eq:lld_extension}
  \minprog{}{\norm{\mat P}_{2\to 2}^* + \gamma \norm{\mat C}_{2 \to \infty}^*}
  { \fronorm{\mat X - \mat P - \mat C} \le \eta}
\end{equation}
When $\eta=0$, this is equivalent to our proposal~\eqref{eqn:lld} for
the gross corruption model.  Other loss functions are also possible.
Note that the Frobenius norm remains invariant under a rotation on the
right, which is a feature of~\eqref{eqn:lld} that we would like to
preserve.  

This formulation is also studied in the independent
work~\cite{Xu2010,Xu2010a}.  It is shown there that under some
technical conditions, the decomposition from~\eqref{eq:lld_extension}
results in a decomposition where $\mat P_\opt$ is close to a matrix
with the same row-space as the true observations, and the matrix $\mat
C_\opt$ is close to a matrix that correctly identifies the column
support of the corruption.

\section{Previous Work}
\label{sec:prevwork}

This section describes previous work on robust formulations of
principal component analysis.  Convex approaches to robust PCA are
unusual, and, as a consequence, many other attempts at robust PCA lack
rigorous algorithms.  Often, proposals are put forward with a
mathematical formulation and only a heuristic algorithm---or an
algorithm without a clear mathematical formulation.

In Sections~\ref{sec:pppca_background} and~\ref{sec:cvx_approach}, we
describe the two methods in the literature most closely related to our
proposals.  We then describe in detail an approach for robust PCA
recommended by Maronna~\cite{Maronna2005} with which we provide
comparisons in Section~\ref{sec:experiments}.  We conclude with a
short overview of other robust PCA proposals that have appeared in the
literature.

\subsection{ Antecedents for MDR: Projection Pursuit PCA}
\label{sec:pppca_background}

Our \mdr\ proposal is a particular instance of an approach that has
come to be known as \emph{projection-pursuit PCA} (PP-PCA), as we
discuss in Section~\ref{sec:pppca}.  The theoretical properties of
PP-PCA are well understood; see for instance~\cite{Cui2003}
and~\cite{Croux2005}.  

All of the algorithms we have found in the literature for computing
PP-PCA are meant to operate with an arbitrary scale.  In view of the
fact that the PP-PCA problem is \NP-hard, it is unsurprising that the
literature appears to contain no PP-PCA algorithms with proofs of
correctness and tractability.  Indeed, we have been unable to find
other work that recognizes that the PP-PCA problem is intractable in a
rigorous sense.

The original study of Li and Chen~\cite{Li1985} uses a Monte Carlo
approach that was found to be computationally expensive.  In theory,
even simple Monte Carlo methods (e.g., randomly sampling the unit
sphere) can produce arbitrarily good solutions to
problem~\eqref{eq:varianceinv} with an arbitrary (continuous) scale.
Given the computational hardness of the problem, it is unlikely that
Monte Carlo approaches can provide guarantees of computational
efficiency.

Current algorithms for PP-PCA rely on heuristics.  A popular and fast
algorithm for generic projection-pursuit PCA is the finite direction
method (FDM) of Croux and Ruiz-Gazen~\cite{Croux2005}.  This technique
replaces the search over the entire unit sphere $\ltwo{\vec v}=1$ with
a finite search over the directions that appear among the
observations: $\vec v \in \{\vec x_1/\ltwo{\vec x_i}, \dotsc, \vec x_n
/\ltwo{\vec x_n}\}$.  The hope is that directions of large scale are
likely to be well approximated by directions appearing in the data.
This heuristic to performs poorly when $n$ and $p$ are large because
it takes an extremely large number of points to cover a
high-dimensional sphere.

% We illustrate the potential for poor behavior in the
% following example.

% \subsubsection{A pathological example for FDM}
% This example illustrates the need for rigorous approximation bounds
% for PP-PCA algorithms.  While the FDM heuristic may hold true in the
% case where the number of observations $n$ is much larger than the
% number of dimensions $p$, we can easily conjure situations where it
% fails quite badly.  Suppose our given data is $\vec x_i = \vec e_i$,
% where $\vec e_i$ is the $i$th coordinate vector.  The finite direction
% method of Croux and Ruiz-Gazen with the MD scale will find the maximum
% scale $\max_i \lone{\vec e_i}=1$.  Taking $\vec s=
% n^{-1/2}(1,\dotsc,1)$ (the unit normal to the simplex) gives an MD
% scale of $\lone{\vec s} = \sqrt{n}$, meaning that the finite direction
% method can not in general approximate the maximum MD scale to within a
% factor of $n^{-1/2}$.

% When we apply \mdr\ to this situation, the optimal value $\alpha_\opt$
% is $\sqrt{n}$ (the reader can check that the matrix of all ones $\mat
% Z_\opt = n\vec s \vec s^\ad$ is an optimal point in
% equation~\eqref{eq:mdrsdp} when $\mat X = \mathbf{I}$). Recalling our
% remark from Section~\ref{sec:rounding} concerning rank one solutions
% to~\eqref{eq:mdrsdp}, we see that our \mdr\ method returns the
% \emph{exact} solution in a situation where the popular FDM heuristic
% fails badly.  

\subsection{A convex approach}\label{sec:cvx_approach}

Recently, a method of Chandrasekaran et~al.~\cite{Chandrasekaran2009}
has been adapted for robust PCA in~\cite{Candes2009}.  This approach
attempts to decompose the data matrix into a sum of a low-rank matrix
and a sparse matrix via the semidefinite program
\begin{equation}
  \label{eq:nplone}
  \minprog{}{\norm{\mat L}_{2\to 2}^* + \lambda \norm{\mat S}_{1\to
      \infty}^*}{\mat L+\mat S = \mat X.}
\end{equation}
The nuclear norm $\norm{\cdot}_{2\to 2}^*$ promotes low rank
and the matrix $\ell_1$ norm $\norm{\cdot}_{1\to \infty}^*$ promotes
sparsity.  We refer to this method as \nplone.  The
works~\cite{Candes2009, Chandrasekaran2009} provide conditions under
which \nplone\ succeeds in \emph{exactly} recovering a low-rank and
sparse component.

This convex approach is principled in the sense that the mathematical
formulation is also algorithmically tractable.  On the other hand, it
lacks an invariance to a change in the observation basis possessed by
all other methods we discuss, including standard PCA.  That is,
applying a rotation $\mat U^\ad \mat U = \mathbf{I}$ to the data
$\widehat{\mat X} = \mat X \mat U$ does not result in a similar
rotation of the decomposition due to the fact that the norm
$\norm{\cdot}_{1\to \infty}^*$ is not invariant under this
transformation.

One may argue that this invariance is inconsequential: in real data,
the particular choice of coordinates has a meaning and outliers may
occur coordinate-wise.  This argument is defensible in domain specific
examples, such as image data that contain
specularities~\cite{Candes2009}.  Nevertheless, PCA is intended to
locate a coordinate basis that explains data more effectively than the
standard basis~\cite{Hotelling1933}.  If this is the analytical goal,
basis invariance is indeed a requisite property.  See
Section~\ref{sec:no2} for an experiment where this lack of orthogonal
invariance in \nplone\ appears to produce unnerving results.

\subsection{Spherical PCA}
\label{sec:sph}

Another approach, known as spherical principal components
(\sph)~\cite{Locantore1999}, rescales the observations to unit
(Euclidean) norm and applies standard PCA to this modified data.  To
implement the \sph\ method, we first compute a normalized matrix
$\widehat{\mat{X}}$.  Each row of $\widehat{ \mat{X}}$ is the
normalized version of the corresponding row of the centered data
matrix $\mat X$, that is $\widehat{ \vec x}_i = \vec x_i /\ltwo{\vec
  x_i}$.  Using the row-normalization operator from~\eqref{eq:bm_opt},
we can express the normalized matrix as $\widehat{\mat X} =
\mathcal{N}(\mat X)$.

The robust components are then defined as the standard principal
components of the rescaled matrix $\widehat{\mat X}$.  Since all of
the observations from the normalized matrix $\widehat{\mat X}$ have
norm one, there are no large magnitude observations that exert an
undue influence on the principal components.

A study by Maronna~\cite{Maronna2005} shows that \sph\ enjoys good
practical performance.  The ease of implementation and relatively good
behavior of \sph\ leads Maronna to suggest it as the default choice
for robust principal component analysis.  As a result, we use \sph\ as
a baseline comparison for the performance of our robust methods in
Section~\ref{sec:experiments}.

\subsection{Other proposals}
\label{sec:otherprop}

Some of the earliest methods for robust PCA compute approximations of
correlation or covariance matrices using robust methods. Gnanadesikan
and Kettenring propose direct robust estimation of the covariance
matrices through robust estimation of the individual
entries~\cite{Gnanadesikan1972}.  This may lead to counterintuitive
results such as non-positive covariance matrices.  An alternative
approach explicitly enforces positive matrices as minimizers of a
functional such as an $M$-estimator~\cite{Devlin1981}; see also the
more recent study~\cite{Croux2000}.

A representative example of robust PCA from the machine learning
community is the work of De La Torre and Black~\cite{DeLaTorre2003}.
They define the robust components as the minimum of a highly
non-convex energy function and attempt to minimize this energy
function using an iteratively reweighted least-squares algorithm
coupled with an annealing step.  No theoretical guarantees of
correctness for the algorithm are provided.

Another recent approach appears in the paper~\cite{Xu2009} of Xu
et~al.  This algorithm randomly removes observations that appear to
have high influence in the current estimate of the principal
components.  The principal component estimate is computed from the
trimmed data.  Xu et~al. are able to establish strong theoretical
properties of their algorithm, including a high breakdown point in the
high-dimensional scaling regime where $n\to \infty$ and $n/p \to c >
0$.

\section{Numerical Experiments}
\label{sec:experiments}

This section provides some numerical examples comparing our proposals
with standard PCA and other robust PCA methods in the literature.  In
Section~\ref{sec:toppc}, we look at the projection of two data sets on
the top robust component.  Section~\ref{sec:bus} repeats a
multiple-component experiment of Maronna~\cite{Maronna2006} with
additional robust methods.  Section~\ref{sec:movielens} contains a
larger experiment, where we calculate the first two components of a
dense matrix with more than twenty million entries.

All of these experiments and algorithms are implemented with
\textsc{Matlab}.  Following the principle of reproducible
research~\cite{Buckheit1995a}, we provide code that reproduces the
exact experiments in this work~\cite{McCoy2010b}.

\subsection{Projection onto the top  component}
\label{sec:toppc}

In this section, we study the robust component methods applied to two
data sets.  The first set is a selection of environmental factors that
may affect the concentration of nitrogen dioxide around Oslo, Norway.
The second example is constructed from standard iris data. In each
case, we examine the spread of the data in the direction of the top
robust component.

\subsubsection{Experimental setup}
For these experiments, we center the data by removing the Euclidean
median from each observation.  The Euclidean median $\widehat{\vec
  \mu}$ is a robust estimate of the center of the data, and is defined
as
\begin{equation}
  \label{eq:eucmed}
  \widehat{\vec \mu} 
  = \argmin_{\vec \mu} \sum\nolimits_{i=1}^n \ltwo{\vec x_i -\vec \mu}.
\end{equation}
Maronna~\cite[Ch. 9 ]{Maronna2006} gives a method to solve this convex
problem for $\widehat{\vec \mu}$.

We project the data onto the top component for each method and compare
the performance of the methods by the \emph{interquartile range}
(IQR), that is, the distance between the $25$th and $75$th percentile
of the projected data.

We apply extract the dominant component from each data set using our
methods (\mdr and \lld), other robust methods (\sph and \nplone), and
standard PCA.  For \mdr, we use $K=94$ rounding trials as discussed in
Section~\ref{sec:rounding}.  We set the \lld\ weight parameter $\gamma
= 0.8\sqrt{p/n}$.  As recommended in~\cite{Candes2009}, we set the
\nplone\ parameter $\lambda = 1/\sqrt{n}$ for the first experiment.
With the iris data in Section~\ref{sec:iris}, we find that $\lambda =
1/\sqrt{n}$ gives a trivial result: no outliers were identified by
\nplone.  Instead, we use the more favorable choice $\lambda =
0.3/\sqrt{n}$.

\subsubsection{Norwegian nitrogen dioxide data}\label{sec:no2}

Our data for this experiment consists of 500 observations of eight
environmental factors around Oslo, Norway, available on the Statlib
archive~\cite{Aldrin2004}.  The variables include the
log-concentration of nitrogen dioxide (NO$_2$) particles, the number
of cars per hour, and the wind speed, as well as several additional
factors useful for predicting the concentration of NO$_2$ particles.

We calculate the top component of the data using each method.  In
Figure~\ref{fig:no2} we plot the projection of the data onto the
direction of these components using a standard box-and-whisker plot.
The whiskers extend either $1.5$ times the IQR beyond the edge of the
box or to the extreme data point.  We consider points that lie beyond
the whiskers outliers.  We give the percentage of outliers and several
order statistics of the data in Table~\ref{tab:NO2}.

Every robust method results in a larger IQR than \pca.  The \mdr\
component finds the largest IQR, and the \lld\ method finds the
smallest IQR among the robust methods.  Except for \nplone, every
method identifies a direction with a relatively large number of
outliers, which indicates that the data has heavy tails.  

The \nplone\ method is unique because it does not identify a direction
of large spread \emph{outside} of the middle $50\%$ of the data.  We
have observed that a random change of the observation basis causes the
\nplone\ component to perform similarly to the \lld\ component.  By
orthogonal equivariance, the results for methods other than \nplone\
are unchanged by a change in the observation basis.  This indicates
that the behavior of the results given by the \nplone\ component is
due to the lack of orthogonal equivariance.

We note that the approximation ratio for the top \mdr\ component is
near optimal at $0.978$. 

% The entire experiment---with every method included---runs in just
% over two seconds on a modest laptop.

\begin{figure}[t]
  \resizebox{\columnwidth}{!}{ \includegraphics{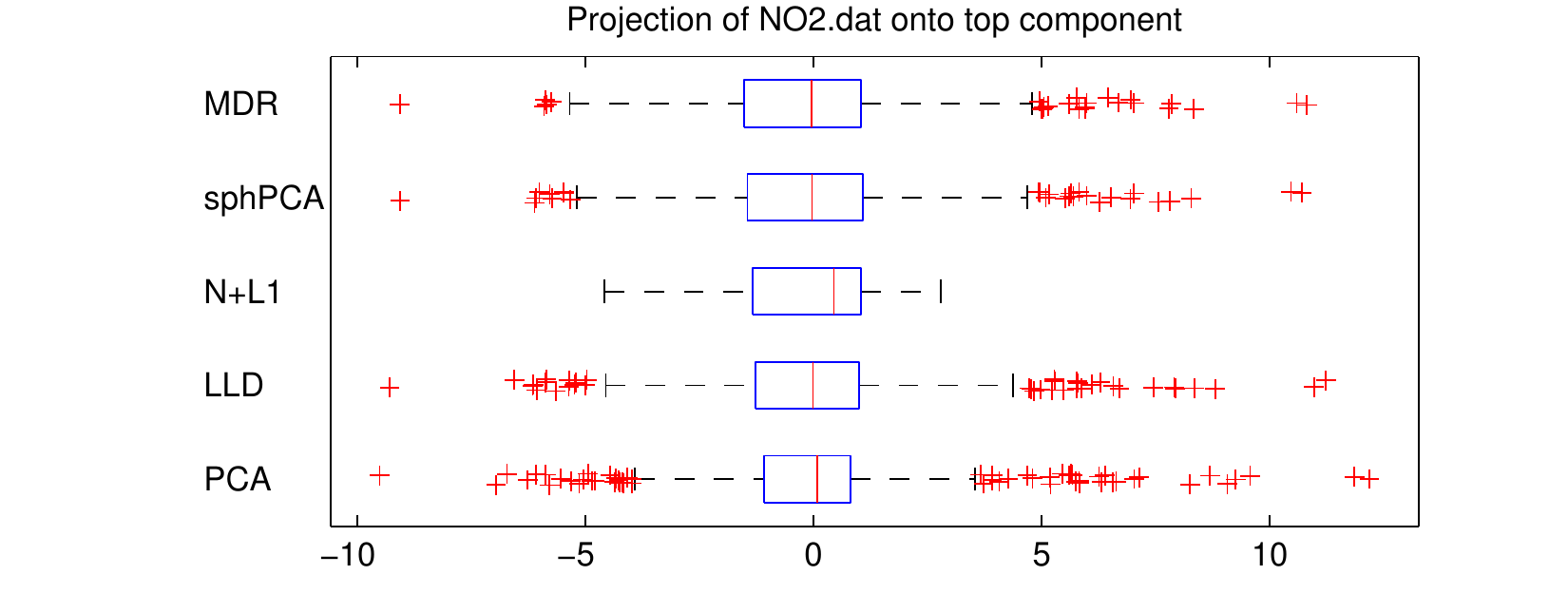}}
  \caption{\label{fig:no2} \textsl{Projection of the Oslo NO$_2$ data
      set onto top components.}  The box surrounds the middle 50\% of
    the data.  The vertical line in the box is the median of the
    data. Each whisker extends either $1.5$ times the length of the
    IQR or to the extreme value of the data, and the red crosses
    beyond the whiskers are the outlying points. The plots are ordered
    by decreasing IQR.  }
\end{figure}

\begin{table}[tbh]
  \caption{\label{tab:NO2} \textsl{Statistics for the projected NO$_{\mathit{2}}$ data.}  The last column 
    lists the percentage 
    of points lying outside the whiskers in Figure~\ref{fig:no2}.}
  \begin{tabular}{|c|c|c|c|c|r@{.}l|r@{.}l|}
    \hline
    Method&	IQR&	min&	25th&	75th&	\multicolumn{2}{|c|}{max} & \multicolumn{2}{|c|}{out}\\
    \hline
    MDR&  	$2.57$&	$-9.07$&	$-1.53$&	$1.05$&	$10$&$82$&	$5$&$00$\% \\
    sphPCA&	$2.53$&	$-9.06$&	$-1.45$&	$1.08$&	$10$&$71$&	$5$&$60$\% \\
    N+L1&	$2.38$&	$-4.58$&	$-1.34$&	$1.05$&	$2$&$79$&	$0$&$00$\%\\
    LLD&   	$2.27$&	$-9.29$&	$-1.27$&	$1.00$&	$11$&$24$&	$7$&$40$\%\\
    PCA&   	$1.89$&	$-9.51$&	$-1.08$&	$0.81$&	$12$&$18$&	$11$&$00$\%\\
    \hline
  \end{tabular}
\end{table}

\subsubsection{Iris data}\label{sec:iris}
We use Fisher's iris data~\cite{Fischer1936} in this experiment.  The
data contains $60$ observations from three different species of iris:
\emph{Iris setosa}, \emph{Iris virginica}, and \emph{Iris versicolor}.
Each observation consists of four measurements, namely sepal length,
sepal width, petal length, and petal width.

Fifty of the observations come from the \emph{setosa} flowers.  We
corrupt these observations with $5$ measurements of \emph{Iris
  virginica} and five measurements of \emph{Iris versicolor}.  We hope
that robust principal components identify a direction of large spread
in the \emph{setosa} bulk of the data.  As a baseline comparison, we
also calculate the dominant principal component of the \emph{setosa}
population without the outlying flowers.

As in Section~\ref{sec:no2}, we project the data onto the direction of
the dominant components.  These points are plotted in
Figure~\ref{fig:iris}; we distinguish the bulk \emph{setosa} points
from the \emph{versicolor} and \emph{virginica} observations.  We
compute an approximate density of the \emph{setosa} observations by
convolving the projected data with a unit volume Gaussian kernel of
width $\sigma = 0.2$.  Table~\ref{table:iris} gives some order
statistics of the projections.

\begin{table}[tbh]
  \caption{\label{table:iris} \textsl{Order statistics for the projection of the 
      \emph{setosa} data onto the top components.}  The last column lists 
    the number of \emph{setosa} points further than 1.5 times IQR 
    left of the 25th percentile or the right of the 75th percentile.
  } 
  \begin{tabular}{|l|r@{.}l|r@{.}l|r@{.}l|r@{.}l|r@{.}l|c|}
    \hline
    Method&	\multicolumn{2}{|c|}{IQR}&	\multicolumn{2}{|c|}{min}&	
    \multicolumn{2}{|c|}{25th}&	\multicolumn{2}{|c|}{75th}&	
    \multicolumn{2}{|c|}{max}& {out}\\
    \hline
    LLD&   	   $0$&$70$&	$-1$&$21$&	$-0$&$41$&	$0$&$29$&	$1$&$14$& $0$.$00$\%\\
    \emph{Setosa} PCA&		$0$&$70$&	$-1$&$22$&	$-0$&$41$&	$0$&$29$&	$1$&$14$& $0$.$00$\%\\
    sphPCA&	   $0$&$69$&	$-1$&$19$&	$-0$&$41$&	$0$&$28$&	$1$&$13$& $0$.$00$\%\\
    N+L1&	   $0$&$66$&	$-1$&$16$&	$-0$&$40$&	$0$&$26$&	$1$&$07$& $0$.$00$\%\\
    MDR&  	   $0$&$37$&	$-0$&$79$&	$-0$&$24$&	$0$&$13$&	$0$&$53$& $0$.$00$\%\\
    PCA&   	   $0$&$19$&	$-0$&$60$&	$-0$&$15$&	$0$&$04$&	$0$&$37$& $6$.$00$\%\\
    \hline
  \end{tabular}
\end{table}

\begin{figure}[tbhp]
  \resizebox{\columnwidth}{!}{ \includegraphics{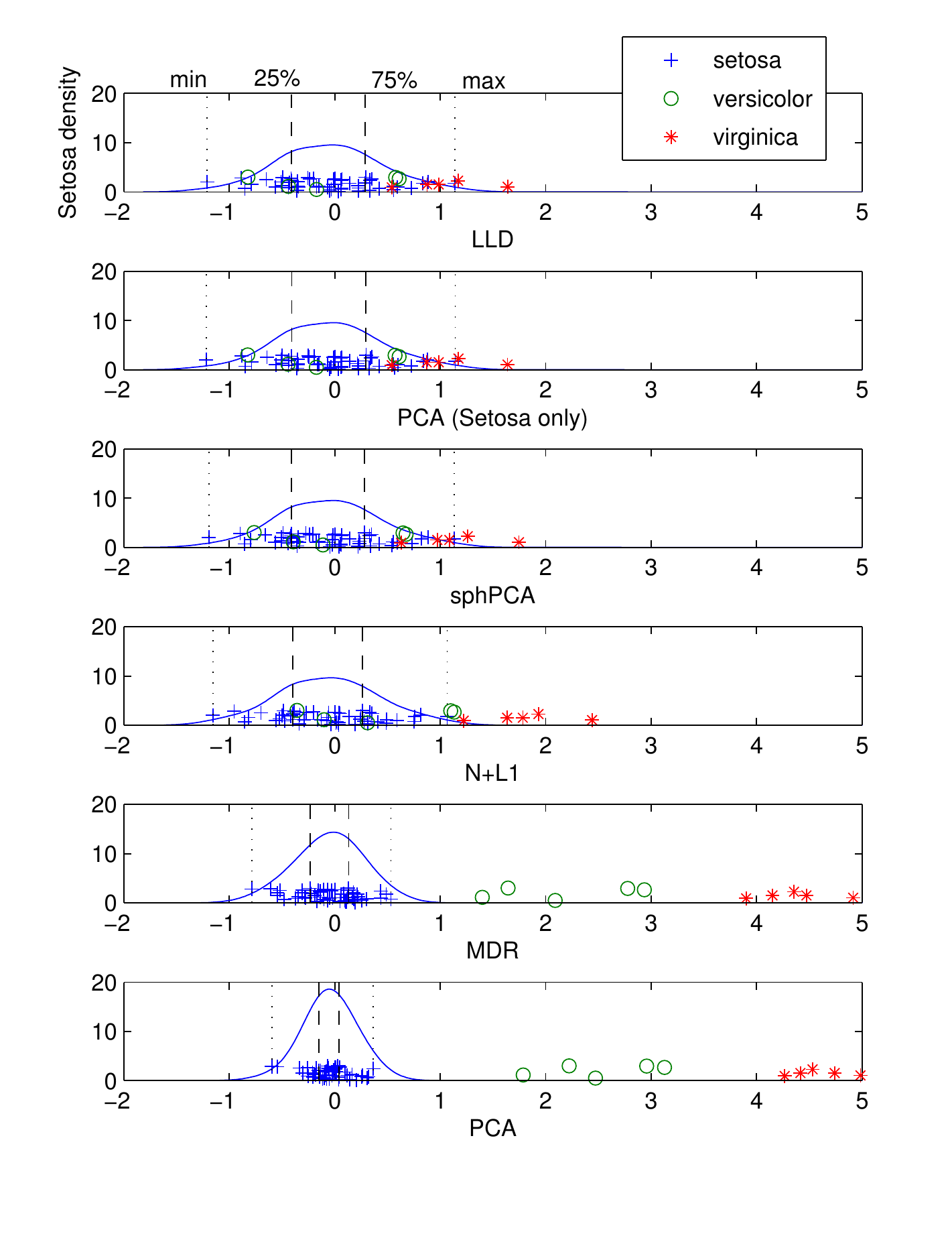}}
  \vspace{-40pt}
  \caption{\label{fig:iris} \textsl{The projections of the iris data onto the
    top components.}  The points are randomly jittered
    above the zero line for readability. The blue curve represents the
    approximate local point density of \emph{setosa}.  Note that \lld\
    and \sph\ essentially provide the same projection as PCA
    \emph{without outliers}.  We sort the plots by decreasing IQR.}
\end{figure}

The dominant component of \lld, \sph, and \nplone\ each achieves an
IQR at least $3$ times that of \pca.  These components do not clearly
distinguish among the three populations, indicating that these methods
are insensitive to the effect of the outliers.  \lld\ and \sph\ appear
the most effective in this situation; indeed, it appears that \lld\
and \sph\ perform as well as \emph{setosa}-only PCA.

Although \mdr\ results in the most modest IQR in the \emph{setosa}
among the robust methods, the IQR associated with the \mdr\ component
is $1.95$ times the IQR of the \emph{setosa} family along the
dominant PCA component.  Unlike the other robust methods, the \mdr\
component discriminates among the three distinct populations.  While
it is clear that \mdr\ \emph{does not} reject the influence of the
outliers, \mdr\ balances the influence of outliers and the bulk of the
data better than \pca.  In this experiment the optimality ratio for
\mdr\ is $0.9975$, certifying that the \mdr\ component is essentially
the direction of maximum mean deviation in the data.

\subsection{Regression Surface for Bus Data}
\label{sec:bus}
In this experiment, we construct a regression surface using multiple
components.  A point is well described by a surface if its Euclidean
distance from the surface is small.  The dominant $T$ classical
principal components span a $T$-dimensional regression surface such
that the sum of the squared distances of the observations to the plane
is minimized.  We would hope that robust components describe the bulk
of the points better than standard components when outliers
contaminate the data.  We illustrate this behavior with an experiment
of Maronna et~al.~\cite[p. 214]{Maronna2006}, which we augment with
additional robust methods.

\subsubsection{Experimental setup}
\label{sec:bus_setup}

Our data consists of $p=18 $ geometric features collected from $n=218$
bus silhouettes~\cite{Siebert1987} that we arrange into an $n\times p$
matrix $\mat X$.  Following Maronna et~al., we remove the $9$th
variable from the data and divide the columns of $\mat X$ by their
median absolute deviation (MADN), a robust measure of scale defined as
\[
\MADN(\vec x) = \median(\abs{\vec x - \median(\vec x)}).
\]
We then center the observations by their Euclidean median.  We
compute the top three components using PCA, \mdr, \lld, \sph, and
\nplone.  We take the \lld\ parameter $\gamma = 0.8\sqrt{n/m}$, the
\nplone\ parameter $\lambda = \sqrt{1/m}$, and the rounding count of
\mdr\ $K=94$.

For each method, we determine the Euclidean distance from each
observation to the orthogonal regression plane spanned by the dominant
three components.  In Figure~\ref{Fig:bus}, we plot the ordered
distances to the robust hyperplanes against the ordered distances to
the PCA hyperplane.

Since the PCA regression surface minimizes the sum of squared
distances to the observations, not all of the observations can lie
below the 1:1 line.  However, a large number of points below the 1:1
line indicates that a robust regression surface explains the bulk of
the data better than the classical surface.

\subsubsection{Discussion}
Figure~\ref{Fig:bus} focuses on the third and fourth quantiles of the
data; the first and second quantiles roughly follow the pattern
apparent in the third quantile.  For clarity, we omit the three most
outlying points that would appear in the upper right corner of the
figure.  Each robust method results in a regression surfaces that
explains the data better than PCA for more than $75\%$ of the points.
In the third quantile, both \nplone\ and \sph\ lose their explanatory
advantage over PCA.  It is not until the after $95\%$ of the data that
\mdr\ and \lld\ provide worse explanations than PCA.  \lld\ is the
dominating method through the latter part of the data.

\begin{figure}[t]
  \resizebox{\columnwidth}{!}{ \includegraphics{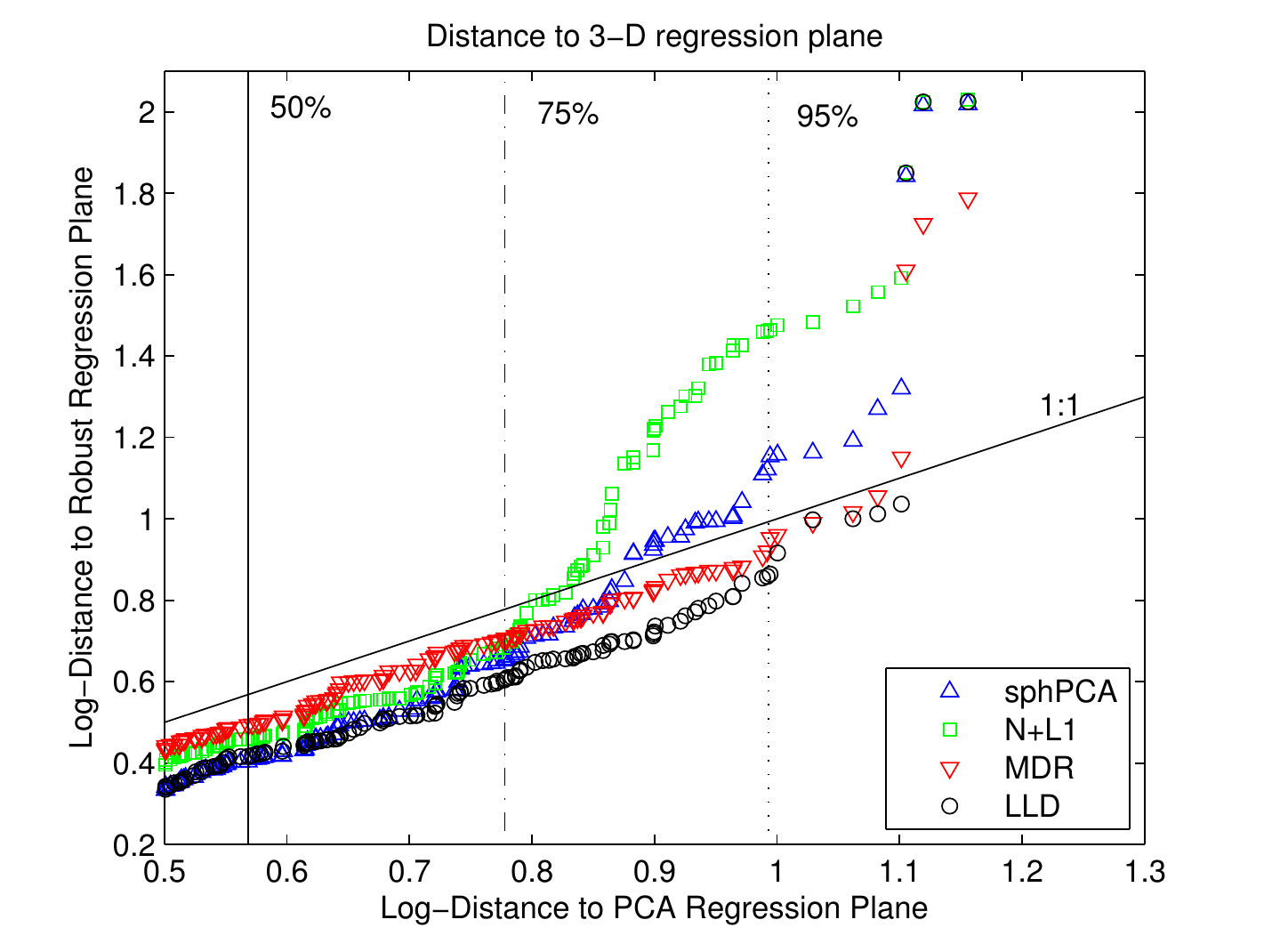}}
  \caption{\textsl{The distance of points to robust regression
      surfaces as a function of the distance of points to the standard
      PCA regression surface.}  The regression surface is determined
    by the top three components from each method.  Points to the left
    of the median follow the same generic pattern as points in the 3rd
    quartile, and are therefore omitted. Three extreme points to the
    right are also omitted.}
  \label{Fig:bus}
\end{figure}

\mdr\ explains the bulk of the data less effectively than the other
robust methods, yet the final outlying observations are explained
better by \mdr\ than the other methods.  This indicates that \mdr\ is
more sensitive to outlying points than the other robust methods, but
is less sensitive to outliers than standard PCA.  The optimality
ratios for the first three \mdr\ components are, respectively,
$0.99999$, $0.99992$, and $0.97253$, implying that \mdr\ essentially
succeeds in PP-PCA with the MD scale for this data.

Finally, we note that changing the \nplone\ parameter to $\lambda =
2\sqrt{1/n}$ results in performance similar to \lld.

\subsection{Movielens}
\label{sec:movielens}
We finish this section with a larger example: the million-rating
movielens data~\cite{movielens}.  The data consist of 6040 users
rating and 3952 movies, though several movies are replicated.  The set
contains just over one million ratings.  Each rating is between one
and five stars, and each user in the data set rated at least 20
movies.

We arrange these responses into an $n=6040$ by $p=3952$ matrix $\mat
X$ whose rows correspond to the users and whose columns correspond to
the movies.  We set unrated movies to the user's median rating, and
center each user's ratings by their personal median.  As with our
other experiments, we center the rows by the Euclidean median, which
results in a dense matrix with nearly $24$ million entries.

We then compute the top two components using PCA, \mdr, \lld, and
\sph.  In order to speed up processing for \lld, we set $\gamma =
\sqrt{100/p}$.  As discussed in Section~\ref{sec:lldalg}, this choice
of $\gamma$ limits the rank of the iterates $\mat P^{(k)}$ in the ALMM
algorithm, which allows us to compute a partial SVD at each step.  Our
choice $\gamma= \sqrt{100/n}$ results in iterates whose rank is
roughly $10$; the rank of the optimal point $\mat P_\opt$ is nine.

Each component $\vec v$ represents a direction in movie coordinates.
The magnitude entry $[\vec v]_i$ indicates how much $\vec v$ points in
the direction of movie $i$.  We use these magnitude of the entries in
the components to rank the movies.  We call movies with large
magnitudes ``important,'' and we call the corresponding entry of the
component a movie's ``importance.''

\subsubsection{Discussion}
\label{sec:movielensdiscussion}

Table~\ref{tab:movielens1st} displays the five most important movies
identified by the first standard principal component, along with the
importance and rank calculated assigned to these movies by the robust
components.  Each method agrees that the violent mobster movie
\emph{GoodFellas} is the most important film. Indeed,
\emph{GoodFellas}, \emph{Army of Darkness}, \emph{A Little Princess},
and \emph{Stand by Me} are ranked in the top five movies by every
method.  However, PCA ranks \emph{Pushing Hands} much higher than the
robust methods.
\begin{table}[tp]
  \centering
  \caption{
    \textsl{Most important movies as given by the  first standard principal 
      component.}  The decimal numbers represent the weight 
    each  component puts on a movie. The integer to the right of 
    the weight is the rank of the movie under the given component.
  }
  \label{tab:movielens1st}
  \begin{tabular}{|l|l|r|l|r|l|r|l|r|}
    \hline 
    Movie & \multicolumn{2}{|c|}{PCA} & \multicolumn{2}{|c|}{MDR} & \multicolumn{2}{|c|}{SPH} & \multicolumn{2}{|c|}{LLD} \\
    \hline
    GoodFellas & $0.0708$& 1 & $0.1016$& 1 & $0.1092$& 1 & $0.0885$& 1 \\
    Army of Darkness & $0.0697$& 2 &$ 0.0914$& 3 & $0.0970$& 4 & $0.0832$& 2 \\
    A Little Princess & $0.0664$& 3 & $0.0899$& 4 & $0.1028$& 2 & $0.0826$& 3 \\
    Pushing Hands & $0.0657$& 4 & $0.0772$& 11 & $0.0827$& 10 & $0.0745$& 8 \\
    Stand by Me & $0.0656$& 5 & $0.0853$& 5 & $0.0896$& 5 & $0.0765$& 5\\
    \hline
  \end{tabular}

  \vspace{12pt}
  \caption{\textsl{Most important movies: second  component.}}
  \label{tab:movielens2nd}
  \begin{tabular}{|l|r|r|r|r|r|r|r|r|}
    \hline 
    Movie & \multicolumn{2}{|c|}{PCA} & \multicolumn{2}{|c|}{MDR} & \multicolumn{2}{|c|}{SPH} & \multicolumn{2}{|c|}{LLD} \\
    \hline
    Nikita & $0.0982$&1 & $0.1099$&2 & $0.0959$&13 & $0.1071$&5 \\
    Citizen Kane & $0.0945$&2 & $0.1051$&4 & $0.0935$&15 & $0.1104$&3 \\
    Fried Green Tomatoes & $0.0917$&3 & $0.0934$&8 & $0.0727$&35 & $0.0944$&11 \\
    Unforgiven & $0.0891$&4 & $0.0923$&10 & $0.0877$&21 & $0.0982$&9 \\
    Mommie Dearest & $-0.0855$&5 & $-0.1108$&1 & $-0.1522$&1 & $-0.1281$&1 \\
    \hline
    \end{tabular}
\end{table}

In Table~\ref{tab:movielens1st}, each importance has positive sign.
For each method, the first component assigns very few movies a
negative importance for the first component.  This fact comes about
because the typical user rating is positive; that is, the sum $\sum_j
[\vec x_i]_j$ is greater than zero for most users.

Table~\ref{tab:movielens2nd} displays the results for the second
components.  Each robust component views \emph{Mommie Dearest} as the
most important movie, while standard PCA relegates it to fifth place.
Neither \emph{Fried Green Tomatoes} nor \emph{Unforgiven} are among
the top five movies for the robust methods.  With the second
component, \sph\ takes the most dramatic shift away from PCA, with
only \emph{Mommie Dearest} making it into the top ten movies.

Of course, rankings are not the whole story.  The signs are very
consistent\footnotemark\ between methods.  \emph{Mommie Dearest} is
negative for every method considered and \emph{Fried Green Tomatoes}
is positive.  The sign consistency indicates that these components are
measuring essentially the same thing.  

The magnitude of the importance are also telling.  PCA assigns the
smallest weight to every movie, with the exception of the second
component of \sph.  This indicates that the robust methods are willing
to assign more importance to discriminating movies.

\footnotetext{Since components are only defined up to a sign, we mean
  that the sign pattern in Tables~\ref{tab:movielens1st} and
  \ref{tab:movielens2nd} are equivalent modulo multiplication by
  $-1$.}

% \subsection{Identifying a Known Subspace from an Elliptical Distribution}
% \label{sec:angles}

% One of the most commonly studied data models in the statistics
% literature is that of elliptically distributed data, including, of
% course, the famous multivariate Gaussian model.  Given data drawn from
% an elliptical distribution, the principal components are ideally
% aligned along the axes of this ellipse since these directions
% correspond to the direction of maximum variance in the true model. Of
% course, the principal components can only approximate these true
% directions of large variance when the number of samples is small. 

% We compare the ability of \lld, \mdr, \nplone, \sph, and PCA to
% recover a subspace close to the theoretically ideal subspace under an
% elliptical distribution.  Our setup is challenging: our ellipse has
% covariance with very slowly decaying singular values, the number of
% dimensions is equal to the number of observations, and the tails of
% our distribution are exponential.  

% Introduce the experiment: expect component analysis to identify
% subspaces corresponding to the underlying model.

% Elliptical data model: uniform direction, exponentially distributed norm.

% Why we chose norm: gaussian data would have norm with mean $p$, but
% with variance on the order of $\sqrt{p}$.

% This setting is essentially designed for spherical principal components.

% Our setup is challenging: slowly decaying singular values.  

% Discuss results.

% Figure.

% Table of median angles.

% \section{Conclusion}
% \label{sec:conclusion}

\section*{Acknowledgments}
\label{ack}
The first author would like to thank Alex Gittens, Richard Chen, and
Stephen Becker for valuable discussions regarding this work.

%\clearpage
\appendix

\section{Proof of Theorem~\ref{thm:mdr}}\label{sec:mdrproof}
This appendix contains the proof of Theorem~\ref{thm:mdr} that we
repeat below as Theorem~\ref{thm:mdrapp}.  We begin with some
supporting results.  The following result of Alon and
Naor~\cite[Sec. 4.2]{Alon2006} allows us to bound the expectation of
$\lone{\mat X \vec v_\opt}$ below.  The essence of this result goes
back to a 1953 paper of Grothendieck~\cite{grothendieck1953résumé};
see also the little Grothendieck theorem in~\cite[Sec. 5b]{Pisier1986a}.
\begin{lemma}\label{lemma:alon-naor}
  Let $\alpha_\opt^2$ be the value of the optimization
  problem~\eqref{eq:mdrsdp} of Algorithm~\ref{alg:mdr}.  Then
  $\alpha_\opt^2 \ge \norm{\mat X \mat X^{\ad}}_{\infty\to 1}$.
  Moreover, let $\vec y^{(k)}$ be one of the vectors generated in
  Step~\ref{step:mdr_random}. Then $\xpcd \ltwo{ \mat X^{\ad} \vec
    y^{(k)}}^2 \ge \frac{2}{\pi}\alpha_\opt^2$.
\end{lemma}
The claim $\alpha_\opt^2 \ge \norm{\mat X \mat X^\ad}_{\infty\to 1}$
also follows from our discussion of the SDP relaxation in
Section~\ref{sec:sdprelax}.  We also need the following proposition.
\begin{proposition}\label{prop:inf1factorapp}
  For each matrix $\mat X$, the identity $\norm{\mat X\mat
    X^{\ad}}_{\infty\to 1} = \norm{\mat X}_{2 \to 1}^2$ holds.
\end{proposition}
\begin{proof}
  We can express
  \[
  \norm{\mat X\mat X^{\ad} }_{\infty \to 1} = \max_{\substack{\norm{\vec
        w}_\infty = 1\\ \norm{\vec y}_\infty = 1}} \Inner{\mat X^{\ad}
    \vec w,\mat X^{\ad}\vec y}.
  \]
  By the conditions for equality in the Cauchy--Schwarz inequality, it
  follows that we can take $\vec w = \vec y$ above.  Hence  
  \begin{equation*}
    \norm{\mat X \mat X^{\ad}}_{\infty \to 1}= \norm{\mat X^{\ad}}_{\infty \to
      2}^2 = \norm{\mat X}_{1\to 2},
  \end{equation*}
  where the last equality is a standard fact concerning adjoint
  operators.
\end{proof}

We use the following variant of the Paley--Zygmund integral
inequality~\cite{Paley1932} to bound the probability that $\lone{\mat
  X \vec v_\opt}$ is less than its expectation.
\begin{lemma}\label{lemma:pz}
  Suppose $Z$ is a random variable such that $0\le Z \le C$ for some
  $C>0$. Then, for any scalar $\theta \in [0,1]$, we have $ \prob(Z >
  \theta \xpcd [Z] ) \ge C^{-1}(1-\theta) \xpcd [Z].  $
\end{lemma}
\begin{proof}
  Split the integral $\xpcd[ Z]$ into two integrals, the first over the
  region $Z \le \theta \xpcd[ Z]$ and the second over the region $Z >
  \theta \xpcd [Z]$. Notice that the former integral is bounded above by
  $\theta \xpcd [Z]$, while the latter integral is bounded above by
  $C \prob(Z > \theta \xpcd [Z])$. Simple algebraic
  manipulation then shows the claim.
\end{proof}

We now restate and prove the main Theorem of Section~\ref{sec:mdr}.
\begin{theorem}\label{thm:mdrapp}
  Suppose that $\mat X$ is an $n\times p$ matrix, and let $K$ be the
  number of rounding trials.  Let $(\vec v_\opt, \alpha_\opt)$ be the
  output of Algorithm~\ref{alg:mdr}.  Then $\alpha_\opt \ge \norm{\mat
    X}_{2\to 1}$. Moreover, for $\theta \in [0,1]$, the inequality
  \begin{equation*}
    \lone{\mat X \vec v_\opt } > \theta \sqrt{\frac{2}{\pi}}
    \alpha_\opt    
  \end{equation*}
  holds except with probability $\mathrm{e}^{-2K(1-\theta^2)/\pi}$.
\end{theorem}
\begin{proof}
  Let $\vec y\in\{\pm 1\}^n$ be a sign vector and define $\vec v =
  \mat X^* \vec y / \ltwo{\mat X^{\ad} \vec y}$.  Then
  \[
  \lone{\mat X \vec v} = \ltwo{\mat X^{\ad}\vec y}^{-1}\max_{\vec w \in \{\pm
    1\}^n}\Inner{\vec w,\mat X \mat X^{\ad} \vec y} \ge \ltwo{\mat X^{\ad}
    \vec y}
  \]
  where the inequality follows by taking the specific choice $\vec w =
  \vec y$.  In particular, this relation implies that the vectors
  $\vec v^{(k)} = \mat X^* \vec y^{(k)}/\ltwo{\mat X^* \vec y^{(k)}}$
  generated in Step~\ref{step:mdr_random} of Algorithm~\ref{alg:mdr}
  satisfy
  \begin{align}\label{eq:key}
    \xpcd \bigl\|\mat X\vec v^{(k)}\bigr\|_1^2 \ge \xpcd \bigl\|\mat
      X^\ad \vec y^{(k)}\bigr\|_2^2 &\ge \frac{2}{\pi} \alpha_\opt^2,
  \end{align}
  where the last inequality follows from the second claim in
  Lemma~\ref{lemma:alon-naor}.
  
  Since $\|{\vec v^{(k)}}\|_2 = 1$, the quantity $\|{\mat X \vec
    v^{(k)}}\|_1^2$ is a positive random variable bounded above by
  $\norm{\mat X}_{2\to 1}^2$. Therefore, inequality~\eqref{eq:key} and
  Lemma~\ref{lemma:pz} imply that
  \begin{align}
    \prob\left(\bigl\|\mat X\vec v^{(k)}\bigr\|_1^2 >
      \theta^2\cdot \frac{2\alpha_\opt^2}{\pi}\right) &
    \ge(1-\theta^2)\frac{2}{\pi}\cdot \left(\frac{\alpha_\opt
      }{\norm{\mat X}_{2\to 1}}\right)^2 \ge
    \frac{2}{\pi}\cdot(1-\theta^2),
  \end{align}
  where we have used the fact that $\alpha_\opt \ge \norm{\mat
    X}_{2\to 1}$ by Proposition~\ref{prop:inf1factorapp} and the first
  claim of Lemma~\ref{lemma:alon-naor}.
  
  In Step~\ref{step:mdr_max} of the algorithm we have chosen $\vec
  v_\opt$ to maximize $\|\mat X \vec v_\opt\|_1^2$, so the inequality
  $\lone{\mat X \vec v_\opt}^2 \le 2(1-\theta^2)/\pi$ holds if and
  only if $\|\mat X \vec v^{(k)}\|_1\le 2(1-\theta^2)/\pi$ for all
  $k$. Therefore, the independence of $\vec v^{(k)}$ for $k =
  1,\dotsc, K$ implies
  \[ 
  \prob\left(\lone{\mat X\vec v_\opt} \le \theta
  \sqrt{\frac{2}{\pi}}\norm{\mat X}_{2\to 1} \right)  \le \left( 1 -
  \frac{2}{\pi}\cdot(1-\theta^2) \right)^K < \mathrm{e}^{-2 K (1-\theta^2) /\pi},
  \]
  which completes the claim.
\end{proof}

\bibliographystyle{abbrv} % (uses file "plain.bst")
%\bibliography{statbib} % expects file "myrefs.bib"
\bibliography{library}

\end{document}

%% file: header.tex
\usepackage{amsmath,amsthm,amssymb}
\usepackage{enumerate}    % Easy control of enumeration numbering
\usepackage{graphics}
\usepackage{multirow}
\usepackage{anysize}     % Allows for any margin settings
\usepackage[pdftex]{hyperref}
\marginsize{1.25in}{1.25in }{1in}{1in}

\makeatletter
\newtheorem*{rep@theorem}{\rep@title}
\newcommand{\newreptheorem}[2]{%
\newenvironment{rep#1}[1]{%
 \def\rep@title{#2 \ref*{##1}}%
 \begin{rep@theorem}}%
 {\end{rep@theorem}}}
\makeatother
\theoremstyle{plain} % Standard
\newtheorem{theorem}{Theorem}%[section] 
\newreptheorem{theorem}{Theorem}
\newtheorem{lemma}[theorem]{Lemma}
\newreptheorem{lemma}{Lemma}

\newtheorem{proposition}[theorem]{Proposition}
\newreptheorem{proposition}{Proposition}
\newtheorem{fact}[theorem]{Fact}

\theoremstyle{definition}

\theoremstyle{remark}

\numberwithin{equation}{section}

\numberwithin{theorem}{section} 
\newcommand{\abs}[1]{\ensuremath{\left\lvert #1 \right\rvert}}
\newcommand{\norm}[1]{{\ensuremath{\left\lVert#1\right\rVert}}}

\newcommand{\lp}[2]{\norm{#2}_{#1}}
\newcommand{\ltwo}[1]{\lp{2}{#1}}
\newcommand{\lone}[1]{\lp{1}{#1}}

\newcommand{\fronorm}[1]{\norm{#1}_{\mathrm{F}}}
\newcommand{\nucnorm}[1]{\norm{#1}_{2\to 2}^*}
\newcommand{\rsnorm}[1]{\norm{#1}_{\mathrm{2\to \infty}}^*}
\newcommand{\rmnorm}[1]{\norm{#1}_{\mathrm{2\to \infty}}}

\newcommand{\eps}{\ensuremath{\varepsilon}} % The usual for ``$\eps > 0$''

\newcommand{\opt}{{\star}}

\newcommand{\Field}[1]{\ensuremath{\mathbb{#1}}}

\newcommand{\Reals}{\Field{R}}
\renewcommand{\Re}{\Reals}
\newcommand{\vect}[1]{\boldsymbol{#1}}
\renewcommand{\vec}[1]{\vect{#1}}
\newcommand{\mat}[1]{\boldsymbol{#1}} 
 
\newcommand{\Inner}[1]{\ensuremath{\left\langle #1 \right\rangle}}

\newcommand{\ad}{\ensuremath{*}} % Adjoint

\newcommand{\NP}{\textsf{NP}}

\renewcommand{\opt}{\star}
\DeclareMathOperator{\minimize}{minimize}
\DeclareMathOperator{\maximize}{maximize}

\newcommand{\minprog}[3]{
  \left.
  \begin{array}{ll} 
    \minimize\limits_{#1} & {#2}  \\
    \text{subject to} & {#3} 
  \end{array} 
  \right. 
} 
\newcommand{\maxprog}[3]{
  \left.
  \begin{array}{ll} 
    \maximize\limits_{#1} & {#2}  \\
    \text{subject to} & {#3} 
  \end{array} 
  \right.
}
\newcommand{\Oh}{\mathcal{O}}

\newcommand{\subg}{\ensuremath{\partial}}

\newcommand{\prob}{\ensuremath{\mathbb{P}}}

\newcommand{\expected}{\ensuremath{\mathbb{E}}}
\DeclareMathOperator{\xpcd}{\expected}
\DeclareMathOperator{\std}{std}

\newcommand{\floor}[1]{\ensuremath{\left\lfloor{#1}\right\rfloor}}

\newcommand{\ggeq}{\succcurlyeq}

\newcommand{\gleq}{\preccurlyeq}

\DeclareMathOperator{\rank}{rank}
\DeclareMathOperator*{\argmin}{arg\,min}
\DeclareMathOperator*{\argmax}{arg\,max}

\DeclareMathOperator{\sign}{sgn}
\DeclareMathOperator{\trace}{trace}

\DeclareMathOperator{\median}{median}
\DeclareMathOperator{\MADN}{MADN}
\DeclareMathOperator{\diag}{diag}

\newcommand{\nplone}{\ensuremath{\mathrm{N+L1}}} % Nuclear + L1 minimization
\newcommand{\lld}{\ensuremath{\mathrm{LLD}}} % Low leverage decomposition
\newcommand{\sph}{\ensuremath{\mathrm{sphPCA}}} % Spherical principal components
\newcommand{\pca}{\ensuremath{\mathrm{PCA}}}
\newcommand{\mdr}{\ensuremath{\mathrm{MDR}}} % Maximum deviation rounding
\newcommand{\md}{\ensuremath{\mathrm{MD}}} 
\newcommand{\pp}{\ensuremath{\mathrm{PP}}}

\DeclareMathOperator{\shrinkrows}{RowShrink}
\DeclareMathOperator{\shrinkspec}{SpecShrink}

\newcounter{algcounter}

  {%
  \end{minipage}%
}